\documentclass[aps, twocolumn,superscriptaddress, longbibliography, nofootinbib, eprint]{revtex4-2}
\usepackage[hidelinks,colorlinks=true,linkcolor=blue,citecolor=blue]{hyperref}
\usepackage{lipsum,adjustbox}
\usepackage{ifthen}
\usepackage{forest}
\usepackage{cancel}
\usepackage{nicefrac}
\usepackage{braket}
\usepackage{dsfont}
\usepackage{bm}
\usepackage{svg}
\usepackage[mode=buildnew]{standalone}
\usepackage{amsthm}
\usepackage{amssymb}
\usepackage{physics}
\usepackage{tikz}
\usepackage[english]{babel}
\usepackage{graphicx}
\usepackage{colortbl}
\usepackage[dvipsnames]{xcolor}
\usepackage[normalem]{ulem}
\usepackage{orcidlink}
\usepackage{soul}
\usepackage[version=3]{mhchem}
\usepackage{siunitx}
\usepackage{amsmath}
\usepackage{algorithm}
\usepackage{algpseudocode}
\usepackage{tikz}
\usepackage{quantikz}
\usepackage{smartdiagram}
\usepackage{setspace}
\usepackage{systeme}
\newtheorem{theorem}{Theorem}

\DeclareMathOperator{\EX}{\mathbb{E}}
\DeclareMathOperator{\Var}{Var}
\newcommand{\ad}{\mathrm{ad}}

\newcommand{\diag}{\operatorname{diag}}
\newcommand{\tA}{\tilde{A}}
\newcommand{\tB}{\tilde{B}}
\newcommand{\tV}{\tilde{V}}
\newcommand{\btheta}{\boldsymbol{\theta}}

\begin{document}
\title
{Gradients, parallelism, and variance of quantum estimates}
\author{Francesco Preti\orcidlink{0000-0002-0343-9049}}\email{f.preti@fz-juelich.de}
\affiliation{J\"ulich Supercomputing Center, Helmholtz AI}
\affiliation{Forschungszentrum J\"ulich, Institute of Quantum Control (PGI-8), D-52425 J\"ulich, Germany}
\author{Michael Schilling\orcidlink{0009-0006-2875-5909}}
\affiliation{Forschungszentrum J\"ulich, Institute of Quantum Control (PGI-8), D-52425 J\"ulich, Germany}
\affiliation{Institute for Theoretical Physics, University of Cologne, Zülpicher Straße 77, 50937 Cologne, Germany}
\author{J\'ozsef Zsolt Bern\'ad \orcidlink{0000-0002-2043-3423}}
\affiliation{Forschungszentrum J\"ulich, Institute of Quantum Control (PGI-8), D-52425 J\"ulich, Germany}
\affiliation{HUN-REN Wigner Research Centre for Physics, Budapest, Hungary}
\author{Tommaso~Calarco\orcidlink{0000-0001-5364-7316}}
\affiliation{Forschungszentrum J\"ulich, Institute of Quantum Control (PGI-8), D-52425 J\"ulich, Germany}
\affiliation{Institute for Theoretical Physics, University of Cologne, Zülpicher Straße 77, 50937 Cologne, Germany}
\affiliation{Dipartimento di Fisica e Astronomia, Università di Bologna, 40127 Bologna, Italy}
\author{F. A. C\'ardenas-L\'opez\orcidlink{0000-0002-2916-2826}}
\affiliation{Forschungszentrum J\"ulich, Institute of Quantum Control (PGI-8), D-52425 J\"ulich, Germany}
\author{Felix Motzoi\orcidlink{0000-0003-4756-5976}}
\affiliation{Forschungszentrum J\"ulich, Institute of Quantum Control (PGI-8), D-52425 J\"ulich, Germany}
\affiliation{Institute for Theoretical Physics, University of Cologne, Zülpicher Straße 77, 50937 Cologne, Germany}
\date{\today}

\begin{abstract}
\hspace{-0.4cm}
Computation of observables and their gradients on near-term quantum hardware is a central aspect of any quantum algorithm. In this work, we first review standard approaches to the estimation of observables with and without quantum amplitude estimation for both cost functions and gradients, discuss sampling problems, and analyze variance propagation on quantum circuits with and without Linear Combination of Unitaries (LCU). Afterwards, we systematically analyze the standard approaches to gradient computation with LCU circuits. Finally, we develop a LCU gradient framework for the most general gradients based on $n$-qubit gates and for time-dependent quantum control gradient, analyze the convergence behaviour of the circuit estimators, and provide detailed circuit representations of both for near-term and fault-tolerant hardware.
\end{abstract}

\maketitle

\section*{Introduction}
The fast developing field of quantum computation requires careful analysis of the sampling complexity of quantum algorithms \cite{Chakraborty2024implementingany}. NISQ quantum algorithms \cite{Bharti_2022}, in particular, can be used to encode specific optimization problems \cite{Kochenberger2014} that depend on classical parameters . In this context, the estimation of fast, reliable gradients of the output of quantum algorithms with respect to classical parameters has been studied extensively \cite{Wierichs2022, bowles2024backpropagationscalingparameterisedquantum, abbas2023quantumbackpropagationinformationreuse}. Quantum algorithms have a vast range of applications.

A quantum algorithm is usually implemented as a family of one or multiple quantum circuits \cite{NielsenChuang2010}, which in turn represent physical experiments on one or more of the available quantum computing platforms, such as superconducting quantum circuits \cite{Krantz2019, Blais2021}, trapped-ions \cite{Haffner2008} or Rydberg atoms \cite{Saffman2010}. In these models, a quantum state is first prepared, evolves under the action of unitary operations and is then measured. Qubits can be measured between unitary operations \cite{Decross2022} (and subsequently reset if needed), so that more complex maps involving mixed states can also be implemented in quantum algorithms. By executing a quantum algorithm multiple times, we can collect data about the possible different outcomes. For example, in the case of variational quantum algorithms \cite{farhi2014quantum, Cerezo2021review, peruzzo2014variational}, the statistics of the measurement process is used to estimate a cost function, which is then optimized with respect to variational parameters using classical optimization methods.

One of the central aspects is the scaling of the number of measurements needed to estimate key circuit observables with precision $\epsilon$. Usually, mean values of arbitrary observables are estimated by sampling from multiple quantum circuits, each one representing an element of an operator basis -- e.g., the Pauli basis. The mean values of the single elements of the operator basis can be evaluated using multiple copies of the same circuit. In the most straightforward implementation the scaling is linear in the number of copies $L$, i.e., $O(L/\epsilon^2)$. 

In shadow tomography models \cite{Huang2020shadow,Huang2021} the scaling can be dramatically improved to reach $O(\log(L)/\epsilon^2)$, whereas using amplitude amplification the scaling becomes sub-linear, at the cost of having to implement the amplitude amplification and the Jordan algorithm routine \cite{Huggins2022, Wada_2025} $O(\sqrt{L}/\epsilon)$. Once the mean values have been estimated, they are summed together with appropriate coefficients. This further increases the variance linearly in the number of terms \cite{Babbush2019}. 

In this paper, we analyze some specific estimators of quantum cost functions. Most of the proposals in NISQ circuits assume the use of a linear combination of measurements for the estimation of observables \cite{Wecker2015, Babbush2019, Rubin2018, consiglio2025variationalquantumalgorithmsmanybody}, which we refer to as Standard Estimator (SE). Such estimator is also used in cases in which mixtures of classical and quantum expectation values need to be computed, using, e.g., quasi-Monte Carlo approaches. For example. QML requires the computation of averages over data sets \cite{Biamonte2017, Jerbi2023, Schuld2019, Schuld2020, schatzki2021entangleddatasetsquantummachine, Preti2024stat}. Quantum control and optimization, on the other hand, for example in the context of so-called robust \cite{schirmer2020robust} or adaptive control/meta-optimization \cite{Preti2022soma, Cervera-Lierta2021}, require to compute averages of cost functions over a certain parameter space \cite{Oshnik2022, Dalgaard2022dup, Preti2022soma, Cervera-Lierta2021, Preti2022purif}, in order to obtain control pulses that are less sensitive to parameter variations. Another example of estimators that use both classical and quantum sampling are (stochastic) parameter-shift rules \cite{LiJun2017, Wierichs2022}, which are used to evaluate gradients of quantum cost functions sampled using variational quantum circuits. Finally, applications of quantum algorithms such as quantum computational fluid dynamics (QCFD) require the (quantum or classical) summation of several estimates from quantum circuits to encode, e.g., the dynamics of relevant partial differential equations on quantum hardware \cite{Jaksch2023}.

Estimation of quantum cost functions can be also performed by implementing linear combinations of unitary operations (LCU) \cite{Somma2002, Childs2012, Childs2017, Cerezo2021review, chowdhury2016quantumalgorithmsgibbssampling, holmes2023nonlinear} on quantum hardware. The question is whether this implementation can be beneficial in specific contexts. In this work, we compare LCU-based estimators to the standard estimator for quantum observables, which uses a different circuit for each non-zero basis element of the observable, and determine the conditions in which the implementation of the former is detrimental or beneficial, i.e., where it provides us with a speed-up over the classical counterpart, limited to when combined with amplitude estimation. We show that the LCU estimator allows for a $\sqrt{L}$ speedup over the standard estimator even for near-term amplitude estimation algorithms, which is in accordance with \cite{Huggins2022, Wada2025}. We explore analytical derivations that confirm partial results and extend their validity to different types of sampling problems. 

Furthermore, we analyze the problem of estimating gradients of quantum cost functions, which has been considered for both variational NISQ circuits and control circuits \cite{Schuld2020, Schuld2019, Wierichs2022, Wiersema2024herecomessun, LiJun2017, Banchi2021measuringanalytic, crooks2019gradientsparameterizedquantumgates, Izmaylov2021, Kyriienko2021, Mitarai2018, li2024efficientquantumgradienthigherorder}, as well as more advanced fault-tolerant algorithms \cite{Gilyen2016}. More specifically, we extend the LCU framework to gradients of multi-qubit gates \cite{Wiersema2024herecomessun} and quantum control problems \cite{Leng2024, kottmann2023evaluating}. This broader LCU framework enables us to differentiate a vast class of quantum cost functions for variational and control circuits. 

The paper is structured as follows. In Section \ref{sec:prob_statement} we discuss the basics of observable and gradient estimation on quantum circuits. Afterwards, in Section \ref{sec:sampling_lcu} we introduce basic estimators of quantum cost functions that make or do not make use of LCU methods and discuss their properties in accordance to the current literature. In Section \eqref{sec:qs} we introduce and outline the properties of amplitude estimation methods and show how their use affect the scaling of previously introduced estimators. We also show how different LCU methods can benefit differently from amplitude estimation compared to the standard LCU procedure. As a paradigmatic example, we also test the estimators on a typical QML regression task. In Section \ref{sec:grad_cost_functions} we introduce the topic of gradient estimation for quantum circuits and review some of the basic methods thereof. We then extend LCU gradient circuits to multi-qubit, multi-parametric quantum gates and quantum control gradients and discuss their relevant scaling. We focus in particular on studying the convergence behaviour of LCU gradients of multi-parametric gates for specific cost functions, such as those discussed in Ref.~\cite{Khatri2019}.

\begin{figure*}[ht!]
\includegraphics[width=\textwidth]{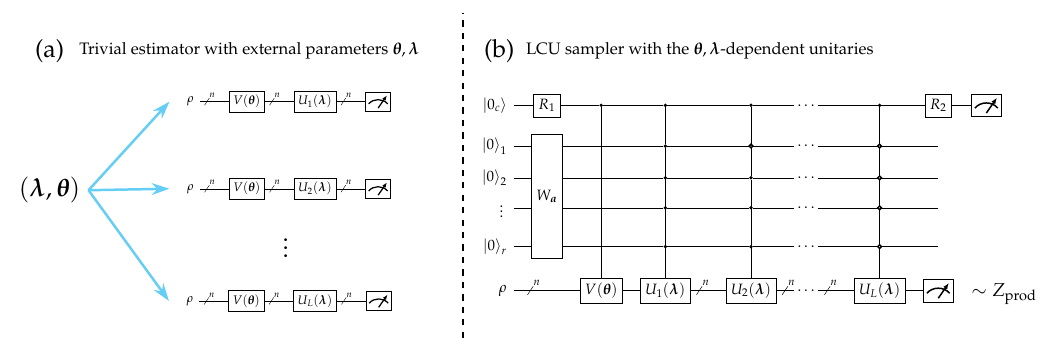}
    \caption{A representation of the two different approaches to observable sampling that are typical of variational quantum circuits: (a) summarizes the Standard Estimator (SE), which prepares $L$ circuits with the same input density matrix and an arbitrary unitary operator $V(\btheta)$. The unitaries $U_1, ..., U_L$ (which are controlled by a parameter vector $\boldsymbol{\lambda}$) prepare, e.g., the different elements of an observable basis or a collection of $L$ non-commuting operators. (b) summarizes the LCU sampler/estimator, which performs the same kind of estimation, but renormalized between, e.g., $I=(1,-1)$. The coefficients of the linear combination of $L$ estimates are computed using classical methods (a) or loaded in the LCU register with $r=\lceil \log(L) \rceil$ qubits using the operator $W_a$ that prepares the state $\ket{a}$ -- see Eqs.~\eqref{eq:W_a} and \eqref{eq:ket_a} -- using a suitable algorithm for state preparation \cite{Iten_2016, da_Silva_2022}. The unitary operations $R_1$ and $R_2$ control the type of cost function to estimate: $R_1 = H$, $R_2 = X$ estimates a cost function as in Eq.~\eqref{eq:h_v_sigma_z}, whereas $R_1 = H$ and $R_2 = H$ estimates a cost function as in Eq.~\eqref{eq:potq_cost} -- see also Ref.~\cite{Somma2002}.} 
    \label{fig:SE_vs_lcu} 
\end{figure*}

\section{Problem statement}
\label{sec:prob_statement}
Sampling from one or multiple quantum circuits 
involves computing linear combinations of binary counts corresponding to different outputs. An example is given by QUBO problems \cite{Kochenberger2014}, where a quantity, which is in a quadratic form with binary arguments, needs to be sampled from various quantum systems. In the case of variational quantum eigensolvers \cite{peruzzo2014variational}, the goal is to minimize the energy of a Hamiltonian given a certain input state. The $n$-qubit observable as a whole is usually not available directly but it can be represented as a linear combination of Hermitian matrices  $P_i$, e.g., Pauli strings, which can potentially be measured using quantum circuits. We consider the observable:
\begin{align}\label{eq:H}
    \mathcal{O} = \sum_{i=1}^L a_i P_i, \quad a_i \in \mathbb{R},
\end{align}
with $L < d^2 = 4^n$.
We limit ourselves w.l.o.g.~ to the case in which $\mathcal{O}$ can be decomposed by considering only one element of the generalized Pauli group, whereby the expression above becomes:
\begin{align}
    \mathcal{O} = \sum_{i=1}^L a_i U_i Z_{\text{prod}} U_i^{\dagger},
\end{align}
where $Z_{\text{prod}} = \underset{i=1}{\overset{n}{\bigotimes}} \sigma_z^{(i)}$ is the $n$-qubit $\sigma_z$ operator and $U_i, i=1,...,L$ are appropriate unitary matrices -- e.g., they map from $Z_{\text{prod}}$ to other elements of the Pauli basis. The choice of $Z_{\text{prod}}$ is arbitrary: another possibility is to map the operator to a single-qubit $\sigma_z$ operator $Z_{\text{prod}}^{(i)} = \mathbb{I} \otimes ... \otimes \sigma_z^{(i)} \otimes \mathbb{I}$ via CNOT operations, but any generalized Pauli operator can be used in principle, as matrices mapping generalized Pauli operator to each other can all be generated using $\text{CNOT}$, Hadamard and Phase gates \cite{Barenco1995, Childs2012}.
The mean value of the observable $\mathcal{O}$ is computed with respect to a density matrix $\rho$, such that for the expected value of $\mathcal{O}$ we can write
\begin{align}\label{eq:h_v_sigma_z}
    \langle \mathcal{O} \rangle = \tr{\rho \mathcal{O}} = \sum_{i=1}^L a_i \tr{\rho U_i Z_{\text{prod}} U_i^{\dagger} }.
\end{align}
Using this representation, we can implement the unitaries $U_i, i=1,...,L$ on different quantum circuits and then measure the register of qubits in the computational basis. We assume that the circuit input state $\rho$ undergoes a parametric evolution generated by a variational unitary. As a result, the expression:
\begin{align}\label{eq:SEst}
\langle \mathcal{O}(\btheta) \rangle = \sum_{i=1}^L a_i \tr{V(\btheta)\rho V^{\dagger}(\btheta) U_i Z_{\text{prod}} U_i^{\dagger} },
\end{align}
encodes an energy minimization problem in up to $L$ different quantum circuits using $N$ real parameters, i.e., $\boldsymbol \theta \in \mathbb{R}^{N}$ and a $n$-qubit variational quantum circuit $U(\btheta) \in \text{V}(d)$, $d=2^n$.
Let us first assume that any two different Pauli strings considered in Eq.~\eqref{eq:H} commute. If this is the case, they can be estimated within the same circuit run, which significantly reduces the amount of measurements needed -- if all $L$ of them commute, estimating their mean values scales as in $O(\lceil \log(L) \rceil /\epsilon^2)$ \cite{Huggins2022}. If they do not commute, up to $L$ circuits need to be executed. Moreover, cost functions for variational circuits are also averaged over additional external parameters, where relevant parameters $\boldsymbol{\lambda} \sim P$ are sampled from a probability distribution $P$:
\begin{align}\label{eq:c_lambda_sampling}
    C(\btheta) = \EX_{\boldsymbol{\lambda \sim P}} \left[ \langle \mathcal{O}(\btheta, \boldsymbol{\lambda}) \rangle \right].
\end{align}
Therefore, there are two types of parameters: \textit{meta-parameters}, denoted by $\boldsymbol{\lambda}$, which are sampled and averaged over -- an example of this is given by Monte-Carlo sampling, where we want to average a value over a data set of parameters -- and \textit{variational parameters}, denoted by $\btheta$ which are generally used for numerical optimization in the context of variational algorithms.

Sampling using Eq.~\eqref{eq:SEst} is not the only option to evaluate the mean value of the observable. 
We can construct an estimator for $\langle \mathcal{O}(\btheta) \rangle$ by first constructing estimators for $L$ different circuits. 
A different estimator can be constructed based on Linear Combination of Unitaries \cite{Somma2002, Childs2012} using a circuit that forks \cite{Park2019} the state evolution in different directions based on controlled operations. Our goal is to analyze the behaviour of such an estimator compared to the standard sequential procedure that uses $L$ circuits. A similar approach can be defined also for gradients of quantum cost functions \cite{Schuld2019}, where the properties of the gate Hamiltonians are exploited to estimate the gradient efficiently. As we discuss later in Section \ref{sec:lc}, this procedure does not really bring any benefit in terms of sampling complexity: on the contrary, it delivers a $\log(L)$ increase in circuit complexity due to the multi-controlled operations. Therefore, in Section \ref{sec:qs} we discuss how to use amplitude amplification to modify the sampling complexity of the estimators and reach an effective speedup using LCU methods.

In the context of variational algorithms, we are also interested not only in the (sampled) cost function $C(\btheta)$, but also in its gradient. Gradients of quantum cost functions can be evaluated in terms of parameter-shift rules, i.e., trigonometric interpolation performed on the quantum cost function \cite{Wierichs2022}. More specific parameter-shift rules can also be determined analytically for several classes of quantum gates \cite{preti2024hybrid} and are expressed as linear combinations of cost function values at different points:
\begin{align}
    \pdv{}{\theta_i} C(\btheta) = \sum_{k=1}^R S_{ik} C(\btheta + \alpha_{ik}\boldsymbol{e}_i),
\end{align}
where $R$ is the number of shifts, $S_{ik}$ and $\alpha_{ik}$ are suitable values that depend on the spectral properties of the gates implemented, and, depending on the type of gate, can be determined analytically \cite{LiJun2017, preti2024hybrid} or numerically \cite{Wierichs2022}.
Another possibility is given by Hadamard-like tests and LCU approaches \cite{li2024efficientquantumgradienthigherorder}, which offer a different solution to the gradient estimation problem. Yet a third approach to gradient estimation, which also finds use also in classical machine learning, is given by Monte Carlo sampling of the cost function gradient \cite{Shakir2020, Sequeira2023}, and also requires the evaluation of a linear combination of cost function samples. As gradient estimation is essentially a special case of estimation of the mean value of a quantum observable, we can make use of the general treatment of LCU vs. standard methods to analyze the sampling complexity of different gradient estimators.

\section{Linear combinations of estimates}\label{sec:lc}

\subsection{Standard Estimator (SE): linear combinations of measurements} \label{sec:sampling_sest}

Our goal is to construct an estimator $\Tilde{C}$ that, using the measurement outcomes collected from the quantum circuits, can successfully approximate $C$ in Eq.~\eqref{eq:c_lambda_sampling}.
Let us consider a collection of circuits numbered $1$ to $L$, each one implementing a unitary $V_1, ..., V_L$ that we use to perform a measurement of $Z_{\text{prod}}$.  We refer to this estimator as the Standard Estimator (SE). This is a straightforward approach in most of the sampling problems in variational quantum circuits \cite{peruzzo2014variational}, so we use this name just for clarity. The principle is simple: we have different circuits that are initialized independently -- see Fig.~\ref{fig:SE_vs_lcu} (a). For each one of these circuits we prepare an identical initial state $\rho$. We first limit ourselves to the case in which the coefficients $a_i$ are all non-negative (which we later generalize in Sec.~\ref{sec:extension_glc}). Formally, we consider first an estimator denoted by the pair $\left(M^{(i)}_{j_1j_2\dots j_n}, \tilde{ C} \right), i=1,...,L$ for a state $\rho$ \cite{Hayashi}, where $M^{(i)}_{j_1j_2\dots j_n}$ are projectors of the form:
\begin{equation}
    M^{(i)}_{j_1j_2\dots j_n}=U_i \left(\Pi_{j_1} \otimes \Pi_{j_2} \otimes \dots \Pi_{j_n}\right)U^\dagger_i,
\end{equation}
where $j_1,j_2,\dots j_n \in \{0,1\}$ and 
\begin{equation}
 \Pi_0=\begin{pmatrix}
     1 & 0 \\
     0 & 0
 \end{pmatrix}, \quad \Pi_1=\begin{pmatrix}
     0 & 0 \\
     0 & 1
 \end{pmatrix}.
\end{equation}
This allows us to estimate $m_i = \tr{ \rho U_i Z_{\text{prod}} U_i^{\dagger}}$ for a given $\rho$ and thus
 \begin{equation}
 \langle \mathcal{O} \rangle_{\rho} = \sum_{i=1}^L a_i \sum_{j_1,j_2, \dots, j_n=0,1}
(-1)^{\sum^n_{l=1} j_l} \tr{\rho M^{(i)}_{j_1j_2\dots j_n}},
\end{equation}
where $a_i$ are the Pauli basis components of the observable given in Eq.~\eqref{eq:H}. Each outcome of a circuit measurement corresponds to a binary string $x^{(i)}_{k_{(i)}} \in {1,...,2^n}$ and is weighted with a coefficient +1 or -1, depending on its binary Hamming weight
\begin{align}
    b(x) = \sum_{l=0}^{n - 1} j_l(x),
\end{align}
where $x=\sum_l j_l(x) 2^{l}$ and $j_l(x)\in\{0,1\}$ are the binary digits of $x$.
Finally, $\tilde{C}$ is the map from the measurement data set to the real line. The SE estimator map is given by: 
\begin{align}
\label{eq:SE_estimator}
    \Tilde{C}_{\text{SE}} =  \sum_{i=1}^L \frac{a_i}{n_s^{(i)}} \sum_{k_i=1}^{n_s^{(i)}} (-1)^{b(x^{(i)}_{k_i})}.
\end{align}
We consider here the estimation of the expected values of $L$ Pauli strings $P_1,...,P_L$. The mean value of a  Pauli string for a quantum state $\rho$ has the following variance:
\begin{align}
    \sigma_{P_i}^2 = \langle P_i^2 \rangle - \langle P_i \rangle^2,
\end{align} 
and due to $P_i^2 = \mathbb{I}$ for any Pauli string, we have:
\begin{align}\label{eq:var_pauli}
    \Var(P_i) = 1 - m_i^2,
\end{align}
where $m_i = \Tr{\rho P_i} \in  \left[-1,1\right]$. Eq.~\eqref{eq:var_pauli} can be considered as the variance of a Rademacher variable, which can be transformed into a (binary) Bernoulli variable with mean $p_i = \frac{1}{2 }(m_i + 1)$ and variance $\Var(P_i) = 4 p_i(1 - p_i)$.
Eq.~\eqref{eq:var_pauli} can be also seen as the variance of a projective measurement $\Pi_i$, with $p_i = \Tr{\Pi_i \rho} \in \left[0, 1 \right]$. Any SE draws measurement values $x^{(j)}_{k_{(1)}}, ..., x^{(j)}_{k_{(L)}}, 1 \leq j \leq n_s^{(i)}$ and $1 \leq i \leq L$ from circuits $V_1,...,V_L$ acting upon the Hilbert space $\mathcal{H} = \mathcal{H}^{(1)} \otimes \mathcal{H}^{(2)} \otimes ... \otimes \mathcal{H}^{(L)}$ with density matrices $\bigotimes_{i=1}^L \rho$. Each circuit is used to estimate the mean value of $Z_{\text{prod}}$ using $n_s^{(i)}$ shots from circuit $i$. Hence, SE has a variance of
\begin{align}\label{eq:var_te}
    \text{Var}(\Tilde{C}_{\text{SE}}) = \sum_{i=1}^L \frac{4a_i^2}{n_s^{(i)}} p_i (1 - p_i).
\end{align}
\begin{figure}[ht!]
    \hspace{-0.5cm}
    \includegraphics[width=8cm]{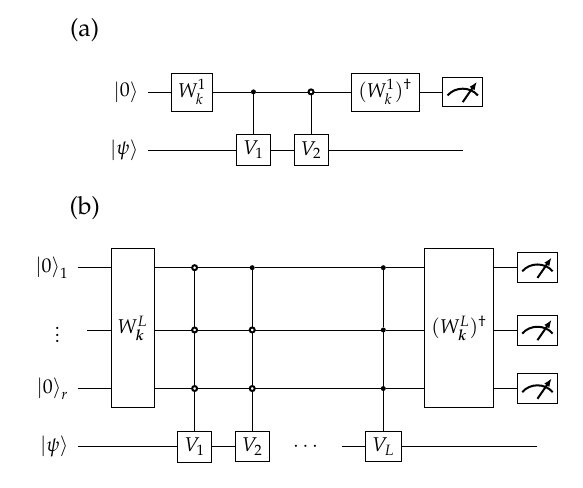}
    \centering
    \caption{(a) Circuit implementing the sum of two unitaries $V_1$ and $V_2$ on a quantum computer using one control qubit and (b) circuit implementing the sum of $L$ unitaries using up to $r = \lceil \log(L) \rceil$ qubits  (both are based on the circuits given in Ref.~\cite{Childs2012}). Upon measuring the control qubit in either $0$ or $1$, the whole state collapses in a state proportional to either $V_1 + V_2$ or $V_1 - V_2$. The LCU can therefore be used to probabilistically implement arbitrary operators acting on a state $\ket{\psi}$, as those found, e.g., in Hamiltonian simulation. In its generalized implementation (b), the LCU generates all possible combinations using coefficients $\boldsymbol{k}=(k_1,...,k_L)^{\text{T}}$ of sums and differences of $L$ unitaries. The linear combination with only positive terms is mapped to the zero state, however the probability of measuring it decreases with $1/L$ \cite{Childs2012}.}
    \label{fig:lcu}
\end{figure}

Due to the absence of entanglement between the (presumed i.i.d.) circuits, there are no correlations between the estimates, so the variance factorizes in the sum of the variances, which can be written as $1 - m_i^2$ for Pauli observables and $p_i(1-p_i)$ for projectors. If we observe that, for all $i=1,...,L$, $n_s^{(i)} \geq \text{min}_{i=1,...,L} \left[n^ {(i)}_s\right] =: n_s $  and use the Chebyshev inequality \cite{wasserman2010statistics}, we can lower-bound the number of shots per circuit as $n_s \geq \frac{L a_{\text{max}}^2}{4\epsilon^2} $, where $\epsilon$ is the precision of the estimation. For a total of $L$ circuits, this results in a circuit sampling complexity of 
\begin{equation}\label{eq:complexitySE}
    S=O(L^2 a_{\text{max}}^2/(4\epsilon^2)),
\end{equation} where $a_{\text{max}} = \text{max}_{i=1,..,L} \left[a_i \right]$ -- see also Refs.~\cite{Wecker2015, Babbush2019, Rubin2018}.

\subsection{Linear Combination of Unitaries}\label{sec:lcu_intro}
The Linear Combination of Unitaries \cite{Childs2012} is a quantum algorithm that allows to implement sums and differences of unitaries on a quantum computer in a probabilistic fashion with a certain depth and success rate depending on the length of the sum to be implemented. In its simplest form it uses the operator:
\begin{align}
    W_{k} = \begin{pmatrix} \sqrt{\frac{k}{k + 1}} & -\sqrt{\frac{1}{k + 1}} \\
    \sqrt{\frac{1}{k + 1}} & \sqrt{\frac{k}{k + 1}}
    \end{pmatrix}.
\end{align}
to create a superposition between control qubit states $\ket{0}$ and $\ket{1}$. Afterwards, conditional operations $V_1$ and $V_2$ are applied on an arbitrary state $\ket{\psi}$, followed by a second operation $W_{k}^{\dagger}$. This leads to a superposition of $V_1 + V_2$ and $V_1 - V_2$ with different probability amplitudes -- see Fig.~\ref{fig:lcu} (a). In particular, we see that the probability of finding the control qubit in its $\ket{1}$ state is given by
\begin{align}
    p_1 = \frac{k}{(k + 1)^2} \norm{(V_1 - V_2) \ket{\psi}}^2 \leq \frac{4k}{(k + 1)^2},
\end{align}
so that the algorithm implements $\frac{1}{\sqrt{a_{+}}}(V_1 + V_2)\ket{\psi}$ for $k \longmapsto \infty$ and $\frac{1}{\sqrt{a_{-}}}(V_1 - V_2)\ket{\psi}$ for $k \longmapsto 0$, where $a_{\pm} = \abs{ \bra{\psi} (V_1 \pm V_2)^{\dagger} (V_1 \pm V_2) \ket{\psi}}$. The algorithm implements one or the other state probabilistically.

If our goal is to apply the sum of $L$ operators, one either applies the circuit represented in Fig.~\ref{fig:lcu} (a) recursively or uses a circuit with multi-controlled gates -- see Fig.~\ref{fig:lcu} (b) and also Appendix in Ref.~\cite{Childs2012}, where typically $\log(L)$ qubits are needed. In this case the success probability will be always smaller than $p_1$. 

The ancillas that implement the summation procedure can also be encoded in a larger number of qubits \cite{Motzoi2017, Araujo2021}. If, e.g., $L$ qubits instead of $\lceil \log(L) \rceil$ qubits are used, the implementation requires only single-qubit controlled gates and no multi-controlled gate. The overall depth of the circuit is then lower \cite{Barenco1995}, but the number of control qubits scales linearly with the number of qubits implementing $V_1$ and $V_2$. The heralded nature of the LCU is particularly useful for quantum simulation. In the case of quantum estimation however, the encoding advantage of using $\lceil \log(L) \rceil$ qubits is traded off with the increased scaling of the variance, as we discuss in the next section.

\subsection{Estimator with Linear Combination of Unitaries (LCU)} \label{sec:sampling_lcu}
Variants of the LCU circuit have been proposed for estimation problems (with some claims of potential speedup) \cite{Park2019, Somma2002}. Our goal here is to characterize the variance properties of these two estimators. We will see that for bounded variables, speedup is not intrinsically possible without further algorithmic improvements. Ref.~\cite{Chakraborty2024implementingany} gives a more in-depth overview of different variants of the LCU algorithm, including continuous variables and integrals. We consider a circuit of the type given in Fig.~\ref{fig:SE_vs_lcu} (b), where the unitary $W_{\boldsymbol{a}}$ generates the state
\begin{align}\label{eq:W_a}
    W_{\boldsymbol{a}} \ket{0} = \ket{a}
\end{align}
on the LCU register, which uses $O(\lceil \log(L) \rceil)$ qubits \cite{Barenco1995}, where 
\begin{align}\label{eq:ket_a}
    \ket{a} = \frac{1}{\sqrt{\norm{\boldsymbol{a}}_1}}\sum_{i=1}^{L} \sqrt{\vert a_i \vert} \ket{i - 1},
\end{align}
for $\boldsymbol{a} = (a_1,\cdots, a_L)^{\text{T}}$ and $\norm{\boldsymbol{a}}_1 = \sum_{i=1}^L \vert a_i \vert$. The weights of the observable can in general be written in vector form $\boldsymbol{a} = \sum_{i=1}^L a_i \boldsymbol{e}_i$, where $\boldsymbol{e}_i$ are unit vectors in $\mathbb{R}^L$.
We refer to the normalized probabilities associated with each amplitude of the state as $w_i$:
\begin{align}\label{eq:w_i_def}
    w_i = \abs{\bra{i} \ket{a} }^2 = \frac{\vert  a_i \vert}{\sum_{l=1}^L \vert a_l \vert}.
\end{align}
It is clear than any convex combination of such weights with coefficients $0 \leq p_i \leq 1$ lies itself between zero and one.
The circuit measures the upper control qubit and the $n$-qubit system (analogously to the case of SE). We calculate the mean and variance of the circuit output measurements. We consider here the case in which measurements are drawn from the computational basis.

\begin{theorem}[Mean and variance of the LCU estimator]\label{theorem:lcu_variance}
The expected value and variance of the observable $\Pi_{\text{LCU}} = \ket{0_c} \bra{0_c} \otimes \mathbb{I}_{L} \otimes \Pi$, where $\Pi$ is an orthogonal projector that describes the measurement operation, are given by
\begin{align}
    &\bar{p} = \sum_{i=1}^L w_i p_i, \label{eq:lcu_p_mean} \\ 
    &\sigma_{\bar{p}}^2 = \sum_{i=1}^L w_i p_i - \left( \sum_{i=1}^L w_i p_i \right)^2 = \bar{p} (1 - \bar{p}), \label{eq:lcu_p_var}
\end{align}
with $p_i = \frac{1}{2}\tr{U_i \rho U_i^{\dagger}\Pi}$. If instead the observable $Z_{\text{LCU}} = \ket{0_c} \bra{0_c} \otimes \mathbb{I}_{L} \otimes Z_{\text{prod}}$ is measured, then the corresponding expected value and variance are:
\begin{align}
    &\bar{m} = \sum_{i=1}^L w_i m_i \label{eq:lcu_sz_mean},\\ 
    &\sigma_{\bar{m}}^2 = 1 - \left( \sum_{i=1}^L w_i m_i \right)^2 = 1 - \bar{m}^2, \label{eq:var_sz_lcu}
\end{align}
with $m_i = \tr{U_i \rho U_i^{\dagger} Z_{\text{prod}}}$. 

\end{theorem}

\begin{proof}
The initial input state of the LCU circuit -- see Fig.~\ref{fig:SE_vs_lcu} -- is given by tensor product of the $n$-qubit input density matrix, the zero state of the LCU register and the zero state of the control qubit:
\begin{align}
    \rho_{\text{in}} = \ket{0_c}\bra{0_c} \otimes \underbrace{\ket{0}\bra{0} \otimes ... \otimes \ket{0}\bra{0}}_{r} \otimes \rho, 
\end{align}
where $r = \lceil \log{L} \rceil$. The evolved density matrix of the LCU circuit is given by 
\begin{align}
    \rho_{\text{out}} = \begin{pmatrix}
        A & B \\
        B^{\dagger} & C \\
    \end{pmatrix},
\end{align}
where $A$, $C$, $B$ are given by \cite{Chakraborty2024implementingany}:

\begin{align}
    A = \sum_{j=1}^L \sum_{i=1}^L \frac{\sqrt{w_i} \sqrt{w_j}}{2} \ket{0_c}\bra{0_c} \otimes \ket{i}\bra{j} \otimes  \rho,
\end{align}

\begin{align}
    B = \sum_{j=1}^L \sum_{i=1}^L \frac{\sqrt{w_i} \sqrt{w_j}}{2} \ket{1_c}\bra{0_c} \otimes \ket{i}\bra{j} \otimes  \left( U_i \rho  \right),
\end{align}

\begin{align}
    C = \sum_{j=1}^L \sum_{i=1}^L \frac{\sqrt{w_i} \sqrt{w_j}}{2} \ket{1_c}\bra{1_c}  \otimes \ket{i}\bra{j} \otimes U_i \rho U_j^{\dagger}.
\end{align}
The entries of $A$, $B$, $C$ may change depending on the values of the single-qubit unitary gates $R_1, R_2$ -- see Fig.~\ref{fig:SE_vs_lcu} -- acting on the first control qubit. The choice of $R_1=H$ and $R_2=X$ leads to the estimation of $\tr{V\rho V^{\dagger} \mathcal{O}}$, while, e.g., the choice of $R_1 = H$ and $R_2 = H$ or $R_2 = SH$ allows us to estimate $\Re{\tr{\rho V  \mathcal{O}}}$ or $\Im{\tr{\rho V  \mathcal{O}}}$, respectively. We consider here the first case, as it resembles Eq.~\eqref{eq:h_v_sigma_z}. The second choice is useful whenever the target cost function is a real (or imaginary) overlap. Hence, the output probability distribution corresponding to the orthogonal projector $\Pi$ acting on the subspace of the density matrix $\rho$ and the measurement line is given by
\begin{align}
    \bar{p} = \Tr{\rho_{\text{out}} \Pi_{\text{LCU}}},
\end{align}
where $\Pi_{\text{LCU}} = \Pi_1 \otimes \mathbb{I}_{L} \otimes \Pi$, $\Pi_1 = \ket{1_c}\bra{1_c}$. The value of $\bar{p}$ is given by
\begin{align}\label{eq:pLCU}
    \bar{p} = \sum_{i=1}^L w_i \frac{1}{2}\Tr{U_i \rho U_i^{\dagger} \Pi} = \sum_{i=1}^L w_i p_i.
\end{align}
Each mean value $p_i$ represents the probability of success of a Bernoulli distribution, and it lies between $0$ and $1$, which introduces constraints on the values $p_i$ that depend on the different unitaries $U_i$:
\begin{align}\label{eq:p1frac2}
    p_i = \frac{1}{2} \Tr{U_i \rho U_i^{\dagger} \Pi}.
\end{align}
The variance of the estimator is given by
\begin{align}
    \sigma_{\bar{p}}^2 = \Tr{\rho_{\text{out}} \Pi_{\text{LCU}}^2 } - \Tr{ \rho_{\text{out}} \Pi_{\text{LCU}}}^2,
\end{align}
which corresponds to the variance of a Bernoulli-type distribution, because $\Pi_{\text{LCU}}^2 = \Pi_{\text{LCU}}$. 

Using Eq.~\eqref{eq:pLCU}, we have 
\begin{align}
    \sigma_{\bar{p}}^2 = \bar{p}(1 - \bar{p}) = \sum_{i=1}^L w_i p_i - \left( \sum_{i=1}^L w_i p_i \right)^2.
\end{align}
The mean value of the estimator is bounded between zero and one, whereas the variance has its maximum at $\sigma_{\bar{x}}^2 = \frac{1}{4}$ when $\bar{p} = \frac{1}{2}$.

Now we turn to the estimation of $Z_{\text{prod}}$ with $R_1 = X$ and $R_2 = \mathbb{I}$.
If instead of the projective measurement, a  measurement of $Z_{\text{prod}}$ is carried out on the $n$-qubit subspace, on the whole Hilbert space we will be dealing with the measurement of the observable $Z_{\text{LCU}} = \Pi_1 \otimes \mathbb{I}_{L}  \otimes Z_{\text{prod}}$. In this case we have $m_i = \Tr{U_i \rho U_i^{\dagger} Z_{\text{prod}}}$ and the mean of the measurement is:
\begin{align}
  \Bar{m} = \sum_{i=1}^L w_i m_i,
\end{align}
as well as the variance
\begin{align}
    \sigma_{\Bar{m}}^2 = 1 - \left(\sum_{i=1}^L w_i m_i \right)^2,
\end{align}
where we used the property $Z_{\text{LCU}}^2 = \Pi_1 \otimes \mathbb{I}_{L}  \otimes \mathbb{I}_n$ and $\sum_{i=1}^L w_i \tr{U_i \rho U_i^{\dagger}} = 1$. We see that the prefactors of the estimates, such as, e.g., $\frac{1}{2}$ in Eq.~\eqref{eq:p1frac2}, heavily depend on the unitaries $R_2$ and $R_1$. Using $R_1 = X$ and $R_2 = \mathbb{I}$, we can also remove the factor $\frac{1}{2}$ also from $\bar{p}$.
\end{proof}
We can define a new estimator using the framework described before: Let $\bar{x}^{(j)}, j=1,...,n_s$ be shots of the LCU circuit that are sampled by measuring the control qubit in $0$ and the corresponding $Z_{\text{prod}}$ operator -- or any other Pauli operator --, then
\begin{align}\label{eq:c_lcu}
    \Tilde{C}_{\text{LCU}} = \frac{\norm{\boldsymbol{a}}_1}{n_s} \sum_{j=1}^{n_s} (-1)^{b(\bar{x}^{(j)})}\bar{x}^{(j)},
\end{align}
with variance
\begin{align}\label{eq:varLCU}
    \text{Var}(\Tilde{C}_{\text{LCU}}) = \frac{\norm{\boldsymbol{a}}_1^2}{n_s} (1 - \bar{m}^2) \leq \frac{\norm{\boldsymbol{a}}_1^2}{n_s}. 
\end{align}

Theorem \ref{theorem:lcu_variance}, however, is not sufficient to bound SE and/or LCU variance, because in principle there could be covariances between the variables. Nonetheless, we show below that these are also bounded similarly to before. 

\begin{theorem}[LCU vs. Classically Correlated Bernoulli Samples]
Consider values $ 0 \leq p_i,w_i \leq 1 $ $\forall i$ with $1 \leq i \leq L$,
$\sum_{i=1}^L w_i = 1$. Let $X_i$ be Bernoulli-distributed random variables with parameter $p_i$, i.e., $\forall i=1,...,L$ we have $\EX[X_i] = p_i$ and $\EX \left[X_i^2 \right] = \EX \left[X_i\right] = p_i$, then we have 
\begin{align}
      \sum_{i=1}^L \sum_{j=1}^L w_i w_j \EX[X_i X_j] \leq 
      \sum_{i=1}^L w_i p_i.
\end{align}
\begin{proof}
By the Cauchy-Schwarz inequality, we know that 
\begin{align}\label{eq:cauchy_schwarz}
    \vert \EX[X_i X_j] \vert  \leq \sqrt{\EX[X_j]} \sqrt{\EX[X_j]} = \sqrt{p_i} \sqrt{p_j},
\end{align}
where we used the fact that $\EX[X_i^2] = \EX \left[X_i\right] = p_i$, and so
\begin{align}\label{eq:gen_mean_bound}
    &\sum_{i=1}^L \sum_{j=1}^L w_i w_k \EX[X_i X_j] \leq \\ \nonumber &\leq \sum_{i=1}^L \sum_{j=1}^L w_i w_j \sqrt{p_i} \sqrt{p_j} =  \left( \sum_{i=1}^L w_i \sqrt{p_i} \right)^2.
\end{align}
Now we employ one of the generalized-mean inequalities \cite{Bullen2003}, i.e., we use the fact that for $p < q$:
\begin{align}
    \left( \sum_{i=1}^L w_i p_i^p \right)^\frac{1}{p} \leq \left( \sum_{i=1}^L w_i p_i^q \right)^\frac{1}{q},
\end{align}
setting $p=\frac{1}{2}$ and $q=1$. After applying this equation to Eq.~\eqref{eq:gen_mean_bound}, we have
\begin{align}\label{eq:cov_comparison}
    \sum_{i=1}^L \sum_{j=1}^L w_i w_j \EX[X_i X_j] \leq \sum_{i=1}^L w_i p_i.
\end{align}\\
\end{proof}
\end{theorem}
Therefore, we have that for such variables $X_1, ..., X_L$ the variance of their linear combination can be bounded from above as follows:
\begin{align}
    \sum_{i=1}^L \sum_{j=1}^L w_i w_j \text{Cov}(X_i, X_j) \leq \sum_{i,j=1}^L w_i w_j p_i (1 - p_j),
\end{align}
where $\text{Cov}(X_i, X_j) = \EX[X_i X_j] - \EX[X_i] \EX[X_j]$ is the covariance. 
The Cauchy-Schwarz inequality provides us with the upper bound for the variance of the linear combination, i.e.:
\begin{align}\label{eq:max_variance}
    \text{Var}\left(\sum_{i=1}^L w_i X_i \right) \leq \left( \sum_{i=1}^Lw_i \sqrt{p_i(1-p_i)} \right)^2. 
\end{align}
If instead we consider shifted estimates $Y_i$ in the interval $I = \left[-1,1 \right]$, i.e., $X_i = \frac{1}{2}(Y_i + 1)$, we obtain that the covariance is always bounded by one, which is the first term in Eq.~\eqref{eq:var_sz_lcu}. Thus, also in this case we have:
\begin{align}
    \sum_{i=1}^L \sum_{j=1}^L w_i w_j \text{Cov}(Y_i, Y_j) \leq 1 - \left( \sum_{i=1}^L w_i m_i \right)^2.
\end{align}
The maximum variance for this case can be derived by applying the same principle as the one used in Eq.~\eqref{eq:max_variance}.

We can see immediately that the variance of the LCU sampler is always larger than the variance of SE in any circumstance.

\begin{theorem}
Both the LCU and SE samplers/estimators yield the same sampling complexity $S=O(L^2 a_{max}^2/{\epsilon^2})$. 
\begin{proof}
For the sampling of a bounded quantity, such the expected value of a Pauli string with respect to variational angles, the variance given in Eq.~\eqref{eq:varLCU} of the LCU sampler is $O(L^2 a_{max}^2/n_s)=O(1/n_s)$ and its sampling complexity is therefore $O(1/\epsilon^2)$. Conversely, the variance of SE for the mean is $O(L a_{max}^2/n_s)=O(L/n_s)$, but $L$ circuits need to be evaluated, so the overall sampling complexity given in Eq.~\eqref{eq:complexitySE} remains $O(1/\epsilon^2)$.
\end{proof}
\end{theorem}

  The two algorithms are therefore equivalent again, although the LCU that uses logarithmic encoding has a slightly -- $\log(L)$ -- deeper circuit. This can be further brought down to $\log(\log(L))$ by using $\log(L)$ additional ancillas \cite{Motzoi2017}. Unless classical correlations are present or are not negligible, some of the applications that have been considered for the LCU circuit \cite{Park2019, Somma2002, qiskit2024} may be useful in specific contexts, but cannot naively provide a speedup with respect to $L$ in terms of sampling complexity (not even in terms of state preparation, as the variance of SE is quadratically smaller, which reduces the total number of shots needed from the $L$ circuits). 

\begin{table*}[]
    \centering
    \begin{tabular}{|c|c|c|c|c|}
        \hline
         & LCU \cite{Childs2012, Chakraborty2024implementingany, Somma2002, Park2019} & SE \cite{Wecker2015, Babbush2019, Rubin2018} & SA-LCU \cite{Chakraborty2024implementingany} & Classical Shadows \cite{Huang2020shadow}  \\
        \hline
        (embarrassing) parallelization & no & yes & yes & yes \\
        \hline
        sampling complexity (i.i.d.) & $O(L^2/\epsilon^2)$  & $O(L^2/\epsilon^2)$ & $O(L^2/\epsilon^2)$ & $O(L\log(L)/\epsilon^4)$\\
        \hline
        sampling complexity (AE) & $O(L/\epsilon)$ & $O(L\sqrt{L}/\epsilon)$ & $O(L^2/\epsilon^2)$ & $O(\sqrt{L} \log(L)/\epsilon^3)$ \\
        \hline 
    \end{tabular}
    \caption{A table representing the differences in terms of sampling complexity of the LCU and SE approaches. In case of normalized sums, all sampling complexities have to be re-scaled by $L^2$ for classical estimation and $L$ for amplitude estimation (AE) \cite{Brassard2002}, or one of its approximate versions \cite{Suzuki2020, Grinko2021} -- see also Section \ref{sec:ae}.}
    \label{tab:samp_comp}
\end{table*}

\subsection{Extension to real linear combinations}\label{sec:extension_glc}
If the coefficients of the linear combination have both positive and negative values, we need to partially modify the encoding defined above and the derivations of the variance scaling. We first observe that we can separate the coefficients in positive and negative terms in the cost function:
\begin{align}
     C = \sum_{i=1}^L a_i m_i.
\end{align}
If we define
\begin{align}
    a_i^{+} = \begin{cases}
        a_i, &  \text{if}  \ a_i \geq 0 \\
        0, &  \text{otherwise},
    \end{cases}
\end{align}
\begin{align}
    a_i^{-} = \begin{cases}
        -a_i, &  \text{if}  \ a_i < 0 \\
        0, &  \text{otherwise}.
    \end{cases}
\end{align}
We define therefore $L^{\pm}$ as the number of coefficients with positive and negative signs respectively, such that $L = L^+ + L^-$. 
In order to be able to construct an estimator similar for Eq.~\eqref{eq:SEst} using only positive weights, we have to first decompose it in its positive and negative terms: 
\begin{align}\label{eq:pos_neg_C}
    C = \sum_{i=1}^{L} a^{+}_i m_i- \sum_{i=1}^{L} a^{-}_i m_i = \norm{\boldsymbol{a}}_1 \sum_{i=1}^L \abs{w_i} (m_i^{+} - m_i^{-}),
\end{align}
where we used Eq.~\eqref{eq:w_i_def} for the definition of $\abs{w_i}$ and $m_i^{\pm}= m_i a_i^{\pm}/\abs{a_i}$. 
We assume now w.l.o.g.~that the positive values appear all before $L^{+}$, i.e., at indices $i=1,...,L^{+}$ and the negative ones after $L^{+}$, i.e., at indices $i=L^{+}+1,...,L$.
Using the formulation of Eq.~\eqref{eq:h_v_sigma_z}, we see that for each one of the positive or negative coefficients we have a unitary $U^{+/-}_i$ for the $i$th positive and negative coefficient, respectively. In order to use the LCU circuit to sample both negative and positive linear combinations, we use the first control qubit line in Fig.~\ref{fig:SE_vs_lcu} (b) to obtain Fig.~\ref{fig:positive_negative_lcu} (a). The unitaries $U^{+}_i, i=1,...,L^{+}$ are controlled on the value zero of the qubit and the unitaries $U^{-}_i, i=L^{+}+1,...,L$ are controlled on value one, whereas the rest of the circuit remains unchanged. This leads to the modified output of the circuit -- $R_1 = H$ and $R_2 = I$:
\begin{align}
    q_0 = \frac{1}{2}\left(\sum_{i=1}^{L^{+}} \abs{w_i} m^{+}_i + \sum_{i=L^{+}+1}^L \abs{w_i} \Tr{\rho Z_{\text{prod}}}\right), \\
    q_1 = \frac{1}{2}\left(\sum_{i=L^{+}+1}^{L} \abs{w_i} m_i^{-} + \sum_{i=1}^{L^{+}} \abs{w_i} \Tr{\rho Z_{\text{prod}}} \right). 
\end{align}
The difference between $q_0$ and $q_1$ estimates the cost function $C$ in Eq.~\eqref{eq:pos_neg_C} that uses positive and negative coefficients up to a bias, i.e., term $\langle Z_{\text{prod}}(0) \rangle = \Tr{\rho Z_{\text{prod}}}$. 
Since empirically using a single LCU register seems to always induce biases, we generalize in the next Section.
An interesting case is one where the number of positive and negative terms in the sum is the same, i.e., $L^{+} = L^{-} = \frac{L}{2}$ and the absolute value of each coefficient is $\abs{w_i} = \frac{1}{L}$. In this case, the bias is the same for both $q_0$ and $q_1$ and is removed when subtracting them. \\
In the general case of uniform coefficients, i.e., $\abs{w_i} = \frac{1}{L}$, we can write the expression as 
\begin{align}
    &q_0 = \frac{1}{2}\left(\frac{1}{L} \sum_{i=1}^{L^{+}} m^{+}_i + \frac{L^{-}}{L} \langle Z_{\text{prod}}(0) \rangle \right), \\
    &q_1 = \frac{1}{2}\left(\frac{1}{L} \sum_{i=L^{+}}^{L} m^{-}_i + \frac{L^{+}}{L} \langle Z_{\text{prod}}(0) \rangle \right),
\end{align} where $ \langle Z_{\text{prod}}(0) \rangle = \tr{\rho Z_{\text{prod}}}$,
which leaves us with 
\begin{align}\label{eq:pnlc_sampling}
    z = 2(q_0 - q_1) + \left(\frac{2L^{+}}{L} - 1\right) \langle Z_{\text{prod}}(0) \rangle,
\end{align}
with $z = \norm{\boldsymbol{a}}_1 C$, which, as expected, reduces to the case $
z =  2(q_0 - q_1)$ for $L^{+} = L^{-} = \frac{L}{2}$. 
\begin{figure*}
    \centering
    \includegraphics[width=0.9\linewidth]{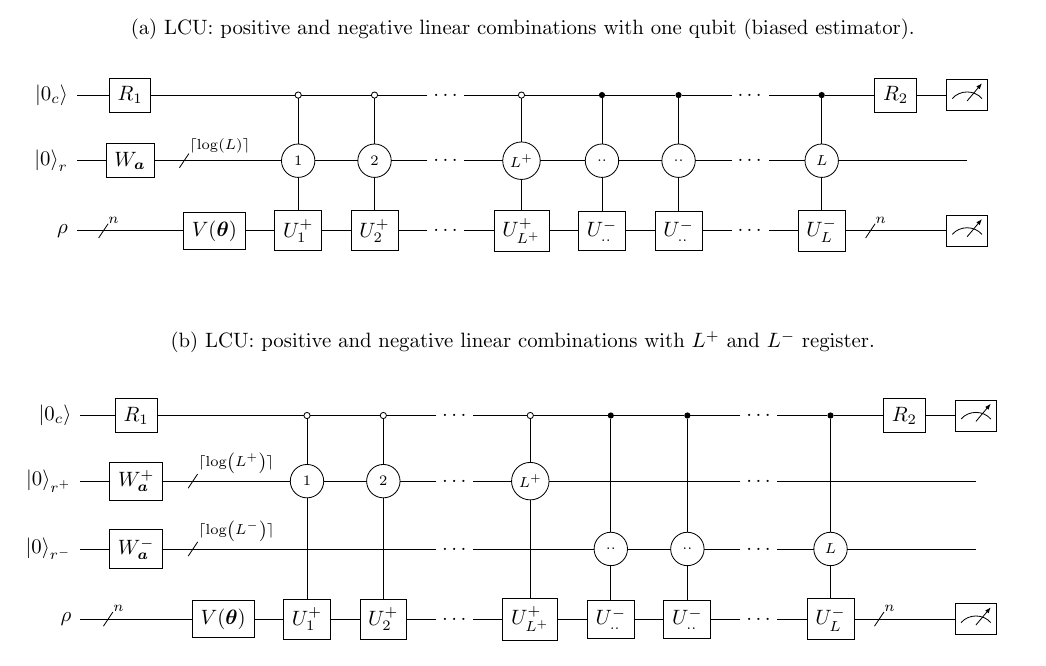}
    \caption{Representation of the circuits used to estimate real linear combinations of estimates. (a) Circuit that uses one single control qubit  and one LCU registerwith $r =\lceil \log(L) \rceil $ and as a result gives a biased estimator -- see Eq.~\eqref{eq:pnlc_sampling}. (b) Circuit that uses one control qubit  and two LCU registers (one for positive and one for the negative terms). Circuit (b) provides us with an unbiased estimator of the real linear combination of estimates -- see Eqs.~\eqref{eq:e0unbiased} and \eqref{eq:e1unbiased}. Here, the unitaries $U^{+}_1, U^{+}_1, ..., U^{+}_{L^{+}}$ and $U^{-}_{L^+ + 1}, U^{-}_{L^+ + 2}, ..., U^{-}_{L}$ refer to the positive and negative signs of the coefficients, respectively. The control values $1, 2, ..., L^{+}$ and $L^{+}+1, L^{+}+2, ..., L$, are implemented using the gates $W_{\boldsymbol{a}}^\pm$ and binary encoding with a total of $r = r^{+} + r^{-}$ qubits, where $r^{\pm} = \lceil \log(L^{\pm}) \rceil$ but other types of qubit encoding for the multi-controlled gates are also possible.}
    \label{fig:positive_negative_lcu}
\end{figure*}
Therefore, this procedure allows us to estimate linear combinations of expected values that contain both positive and negative coefficients.
\subsubsection{Unbiased estimator}
Now, let us analyze the variance in the case of this more general kind of linear combination of estimates. The bias will affect the variance of the estimator with a quadratic factor \cite{Kay1993} $(2L^{+}/L - 1)^2 \leq 1$, which depends on the problem considered. The alternative to this biased estimator is an estimator that uses two LCU circuits, one for the positive and one for the negative coefficients and subtracts their average outcome. 
If we instead use a circuit with two LCU registers instead -- see Fig.~\ref{fig:positive_negative_lcu} (b) --, we can encode positive and negative coefficients $w^{\pm}_i$:
\begin{align}
    w^{\pm}_i = \frac{a_i^{\pm}}{\sum_{i=1}^{L^{\pm}} a_i^{\pm}} = \frac{a_i^{\pm}}{\norm{\boldsymbol{a}^{\pm}}_1},
\end{align}
in the first and second register with $r^{\pm} = \lceil \log(L^{\pm}) \rceil$ qubits respectively, whose unitaries are defined as $W_{\boldsymbol{a}}^+$ and $W_{\boldsymbol{a}}^-$, s.t. $\ket{a^{\pm}} = W_{\boldsymbol{a}}^{\pm} \ket{0}_{L^{\pm}}$. For $R_1=H$ and $R_2=I$, if 0 and 1 are the outcome of the control qubit, the circuit estimates are:
\begin{align}
    &e_0 = \frac{1}{2} \sum_{i=1}^{L^{+}},\abs{w_i^{+}} m_i^{+}  \label{eq:e0unbiased} \\
    &e_1 = \frac{1}{2} \sum_{i=L^+ + 1}^{L} \abs{w_{i}^{-}} m_i^{-}, \label{eq:e1unbiased}
\end{align}
and, as a result, $C = 2(\norm{\boldsymbol{a}^{+}}_1 e_0 - \norm{\boldsymbol{a}^{-}}_1 e_1)$.

\section{Quantum sampling and amplification}\label{sec:qs}

\subsection{Amplitude amplification and estimation}\label{sec:ae}
Amplitude amplification (AA) \cite{Brassard2002} describes a class of algorithms based on Grover search \cite{Grover1996} that allows to perform faster sampling on quantum computers \cite{cui2024quantummontecarlointegration, nannicini2025quantumalgorithmsoptimizers}. If applied to estimation problems in the context of quantum circuits, it is often referred to as amplitude estimation (AE) \cite{Brassard2002} and can be used to sample the information hidden inside of amplitudes of quantum states of the type:
\begin{align}\label{eq:AEst}
    \ket{\psi} = \mathcal{V}\ket{0}_{n+1} = \sqrt{p} \ket{1} \ket{\psi_1} + \sqrt{1 - p} \ket{0} \ket{\psi_2},
\end{align}
where $p$ is a positive value that needs to be estimated and $\ket{\psi_1}$, $\ket{\psi_2} \in \mathbb{C}^{2^n}$ and $\ket{\psi} \in \mathbb{C}^{2^{n+1}}$ are general $n$-qubit quantum states. There exists a quantum algorithm that outputs to an estimator $\Tilde{p}$ of $p$ \cite{Brassard2002}, which is bounded from above as follows given two integers $k$ and $n_q$, where $k$ is related to the success probability and $n_q$ is the number of oracle queries:
\begin{align}\label{eq:amplitude_estimation}
    \abs{p - \Tilde{p}} \leq 2\pi k \frac{\sqrt{p ( 1 - p)}}{n_q} + \frac{\pi^2k^2}{n_q^2}.
\end{align}
The AE routine succeeds with probability $p_{\text{succ}} = 8/\pi^2$
if $k=1$ and $p_{\text{succ}} \geq 1 - \frac{1}{2(k-1)}$ if $k > 1$. The algorithm that performs multiple applications of the Grover operator for AE is simply given by \cite{Brassard2002}:
\begin{align}\label{eq:Q_operator}
    \mathcal{Q} = -\mathcal{V} S_0 \mathcal{V}^{-1} S_{\chi},
\end{align}
where $S_{\chi}$ is the reflection operator that switches the sign of the state if the state is equal to the target, i.e., the state $ \sqrt{p} \ket{1} \ket{\psi_{1}}$ in Eq.~\eqref{eq:AEst} (also known as the \textit{good} state) \cite{Brassard2002, Suzuki2020} and the operator $S_0$ flips the sign of the zero state. Generally, the routine for AE is defined for pure states and not for mixed inputs, such as those used in DQC1 (Deterministic Quantum Computation with one clean qubit \cite{Knill1998}), a quantum complexity class which is closely related to LCU circuits -- see also Fig.~\ref{fig:dqc1-rev} (a) and (b) and Appendix \ref{sec:samplingtraces} for a discussion about DQC1. However, the routine that generates the mixed state can be expressed as tracing out a pure state of higher dimensionality than the one considered as input to the sampling circuit. As a result, to apply AE to one of our problems, we have to consider the entire algorithm, the one that generates the mixed state and the one that generates the state in Eq.~\eqref{eq:U_aest}. As an example and as described in Ref.~\cite{Khatri2019}, the $n$-qubit maximally mixed state $\rho = \frac{\mathbb{I}}{d}$ can be generated by creating a $2n$-qubit pure state
\begin{align}\label{eq:bell_gen}
    \ket{\phi^{+}} = \frac{1}{\sqrt{d}} \sum_{i=0}^{d-1} \ket{i} \otimes \ket{i},
\end{align}
and tracing out one of the $n$-qubit subsystems. So the $(n+1)$-qubit DQC1 circuit can be written as a $(2n + 1)$-qubit circuit by including the generation of the mixed state in its original routine. 

If we consider NISQ circuits, which are not fault-tolerant, we quickly realize that the depth of the Grover's algorithm, which lies at the core of the AE routine, cannot be easily implemented on such circuits. However, alternative versions of AE can be potentially applied on noisy quantum devices \cite{Suzuki2020, Grinko2021, Plekhanov2022variationalquantum}. AE can be used to estimate means of data sampled from arbitrary distributions \cite{Shyamsundar2021} and has a potentially vast range of industry-relevant applications, e.g., in finance. In this context, it can also be used to sample quadratically faster from stochastic processes \cite{Stamatopoulos2020optionpricingusing}.  

\subsubsection{SE with amplitude amplification}\label{ae:SE}
We consider here SE for the trace estimation problem. In this case we need to apply the AE scheme to each SE circuit with index $i=1,...,L$, after encoding the estimation problem appropriately. It is clear that for each term this bound is quadratically better than classical quasi-Monte Carlo sampling. We assume that each $n$-qubit quantum circuit corresponding to each SE term can be written as \cite{Suzuki2020}:
\begin{align}\label{eq:U_aest}
    \mathcal{V}_i \ket{0}_{n+1} = \sqrt{p_i} \ket{1} \ket{\psi_{i,1}} +  \sqrt{1 - p_i} \ket{0} \ket{\psi_{i,2}},
\end{align}
for $i=1,...,L$ and for two quantum states $\ket{\psi_{i,1}}$ and $\ket{\psi_{i,2}}$  that are determined by $\mathcal{V}_i$. This is always possible by extending the original circuit Hilbert space using an additional ancilla qubit. 

For each amplitude encoded in a quantum circuit $\mathcal{V}_i$ we have:
\begin{align}\label{eq:amplitude_estimation_i}
    \abs{p_i - \Tilde{p}_i} \leq 2\pi k \frac{\sqrt{p_i( 1 - p_i)}}{n_q^{(i)}} + \frac{\pi^2k^2}{[n_q^{(i)}]^2},
\end{align}
where $\tilde{p}_i$ is the QAE estimator of $p_i$ and $n_q^{(i)}$ are the queries for the $i$th AE problem.  Our aim is to achieve an overall precision $\epsilon$ of the full estimation. For $i=1,...,L$, $n_q^{(i)} \geq \text{min}_{i=1,...,L} \left[n^ {(i)}_q\right] =: n_q $, we have:
\begin{align}\label{eq:epsilon_ae}
    &\epsilon = \vert \sum_{i=1}^L p_i - \sum_{i=1}^L \Tilde{p}_i \vert \overset{\Delta}{\leq} \sum_{i=1}^L \vert p_i - \Tilde{p}_i \vert  \leq \\ \nonumber &\sum_{i=1}^L \frac{2 \pi k}{n_q}  \sqrt{p_i ( 1 - p_i)} + \frac{k^2L\pi^2}{n_q^2},
\end{align}
where $n_q$ is the number of oracle calls for each circuit $i=1,...,L$ and $k \in \mathbb{N}_{>0}$ and $\Delta$ denotes the use of the triangle inequality. We use here balanced weights ($a_i = 1$ for $i=1,...,L$), but the result can be generalized straightforwardly to the weighted case. The estimation is probabilistic, i.e., it is successful with probability $p_{\text{succ}}$. In the next steps we describe the working principles of the AE routine, see also Ref.~\cite{Brassard2002}. The probability of success of the AE algorithm can be bound from below using the trigamma function $\psi^{(1)}(k) = \sum_{s=k}^{\infty} 1/s^2$ \cite{Brassard2002, AbramowitzStegun}. For any $k \in \mathbb{N}_{>0}$ we have:
\begin{align}
    p_{\text{succ}} \geq 1 - \frac{1}{2}\psi^{(1)}(k).
\end{align}
As the trigamma function fulfills the identity \cite{Brassard2002}:
\begin{align}
    \psi^{(1)}(k) \leq \frac{1}{k - 1},
\end{align}
the lower bound for the success probability can be written as
\begin{align}
    p_{\text{succ}} \geq 1 - \frac{1}{2(k-1)},
\end{align}
where $p_{\text{succ}}$ is the success probability of one quantum AE algorithm running on one of the estimation problems. For $k=1$ we have $p_{\text{succ}} = \frac{8}{\pi^2}$ and considering that the variance $p_i(1 - p_i)$ is always smaller than $\frac{1}{4}$ for $0 \leq p_i \leq 1$, we can use Eq.~\eqref{eq:epsilon_ae} to find a bound for the precision $\epsilon$ of the estimation:
\begin{align}
    \epsilon \leq \frac{\pi L}{n_q} + \frac{ \pi^2 L}{n_q^2},
\end{align}
and after solving the quadratic equation for small $\epsilon$ we obtain for the bound $\frac{\pi L}{n_q} \approx \epsilon$. Naively, the upper bound for the number of evaluations needed seems to be $O(L^2/\epsilon)$ -- since we additionally consider the evaluation of $L$ parallel AE routines -- but a more detailed treatment -- see Appendix \ref{appx:SE_AE} for the derivation from the point of view of Maximum Likelihood Quantum Amplitude Estimation (MLQAE) \cite{Suzuki2020} -- leads to a better bound which matches previous results obtained in other similar contexts \cite{Huggins2022}:
\begin{equation}
    S=O(L\sqrt{L}/\epsilon),
\end{equation} that is, $L$ circuits with $n_q$ queries calls each, in analogy with SE -- to have convergence with probability $p_{\text{succ}}$. 

\subsubsection{LCU with amplitude amplification}

In the case of the circuit shown in Fig.~\ref{fig:SE_vs_lcu}, an amplification oracle can be constructed using the methodology given in Ref.~\cite{Wada2024quantumenhancedmean, Suzuki2020}, which uses an ancilla qubit to encode the estimation of arbitrary observables. 
In classical estimation, the renormalization factor causes the variance of the LCU estimator to grow quadratically with the number of estimates. In the case of AE, it only increases the error linearly. 
\begin{align}\label{eq:U_aest_LCU}
    \mathcal{V}_{\text{LCU}} \ket{0}_{b} = \sqrt{\bar{p}} \ket{1} \ket{\psi_{\text{LCU},1}} +  \sqrt{1 - \bar{p}} \ket{0} \ket{\psi_{\text{LCU},2}},
\end{align}
where $b = (n+1)r$, $\bar{p} = \sum_{i=1}^L w_i p_i$
and for two quantum states $\ket{\psi_{\text{LCU},1}}$ and $\ket{\psi_{\text{LCU},2}}$ that are defined by the action of $\mathcal{V}_{\text{LCU}}$.
An amplified LCU circuit $\tilde{C}^{\text{AE}}_{\text{LCU}}$ estimates the same value as $C_{\text{LCU}}$ given in Eq.~\eqref{eq:c_lcu}, but with a lower variance:
\begin{align}\label{eq:amplified_lcu}
    & \hspace{0.5cm} C = \EX[ \tilde{C}_{\text{LCU}}] = \EX[\tilde{C}^{\text{AE}}_{\text{LCU}}]\\
    &\Var(\tilde{C}^{\text{AE}}_{\text{LCU}}) = \norm{\boldsymbol{a}}_1^2 \frac{\bar{p}(1 - \bar{p})}{n_q^2} = \epsilon^2,
\end{align}
where $C_{\text{full}}$ is the corresponding non-renormalized cost function and $n_q$ the number of queries. As usual the variance of the estimator determines the error threshold $\epsilon$. The number of samples/queries is thus quadratically smaller than in standard sampling and as such $n_q \sim  O(\frac{\norm{\boldsymbol{a}}_1}{\epsilon}) = O(L/\epsilon)$.
We immediately see that the sampling complexity of this estimator scales as $O(L/\epsilon)$ for the non-renormalized estimation. 
We can briefly analyze a further important difference between the two algorithms. In SE case, we run $L$ different AE routines, each one with a success probability $p_{\text{succ}} \geq \frac{8}{\pi^2}$.  In the LCU case, we apply the routine to one single circuit. Clearly, the second approach has an advantage in terms of overall probability of failure. In fact, in the former case there are $L$ independent algorithmic runs that can return the wrong outcome. In the latter case, however, only one circuit with a corresponding AE routine and success probability $p_{\text{succ}}$ is used.

\subsubsection{Single-ancilla LCU}
Estimating a quantum cost function that depends on parameters $\btheta$ requires projective measurements of a density matrix whose dynamics is described by, e.g., a variational circuit. The estimation process can be accomplished by preparing independent circuits or by means of the LCU approach, which requires ancilla qubits to encode the linear combination. A third approach, the single-ancilla LCU (SA-LCU), is presented in Ref.~\cite{Chakraborty2024implementingany}. It uses random sampling of unitaries from the set of weights that are used in the the linear combination. Ref.~\cite{Chakraborty2024implementingany} shows that it provides us with the correct result in the case of projective measurements. In this case, given a collection of unitaries $\mathcal{U}_S = \{U_1, ... U_L\}$ and normalized weights $\{w_1, ..., w_L\}$, such weights induce a probability distribution $\pi_{\mathcal{U}}$ over the set of unitaries \cite{Chakraborty2024implementingany}. 
Randomly sampling a unitary according to the distribution of weights and implementing it on the circuit in Fig.~\ref{fig:SE_vs_lcu} without using the linear combination register returns the expected value:
\begin{align}\label{eq:SA-LCU}
    \bar{p} = \EX_{U \sim \pi_{\mathcal{U}}} \left[p(U) \right] = \sum_{i=1}^L w_i p_i,
\end{align}
where $p_i$ is the probability of finding the $i$th state, i.e., $U_i \rho U_i^{\dagger}$ in the eigenstate of $\Pi$. 
This is equivalent to preparing a mixed state $\rho_{\text{LCU}}$ of the type:
\begin{align}
    \rho_{\text{LCU}} = \sum_{i=1}^L w_i U_i \rho U_i^{\dagger},
\end{align}
and performing a measurement $\Pi$ on it. In other words, if it is possible to efficiently generate this mixed state, then one can obtain the same kind of parallelization as with LCU. This approach resembles the one used for trace estimation in DQC1, where we prepare instead a maximally mixed state $\rho = \frac{\mathbb{I}}{d}$. If we assume a large number of samples and only consider the variance with respect to the sampled random unitaries \cite{Mele2024introductiontohaar}, we obtain the expression:
\begin{align}\label{eq:var_lcu_ancfree}
    \Var_{U \sim \pi_{\mathcal{U}}} \left[p(U)\right] = \EX_{U \sim \pi_{\mathcal{U}}}  \left[p(U)^2 \right] - \EX_{U \sim \pi_{\mathcal{U}}}  \left[p(U) \right]^2,
\end{align}
which gives
\begin{align}\label{eq:no_tot_var_lcu}
    \Var_{U \sim \pi_{\mathcal{U}}} \left[ p(U) \right] = \sum_{i=1}^L w_i p_i^2 - \bar{p}^2.
\end{align}
However, this derivation is not valid in the few-shot limit, as we need to consider the nested sampling of both $U$ and the shots coming from the quantum circuit that we use to estimate $p(U)$. In the latter case, the variance is computed using the law of total variance \cite{soch2024statproofbook}. By defining $\bar{x}$ as the random variable associated with a single shot from the circuit with unitary $U$, we have
\begin{align}\label{eq:sa_lcu_full}
    \Var(\bar{x}) = \EX_{U \sim \pi_{\mathcal{U}}}\left[ \Var_q(\bar{x}) \right] + \Var_{U \sim \pi_{\mathcal{U}}} \left[ {\EX_q[\bar{x}]} \right],
\end{align}
where $\Var_q(\bar{x})$ and $\EX_q[\bar{x}]$ refer to the quantum statistics of the single-shot circuit sampling. Eq.~\eqref{eq:sa_lcu_full} results in:
\begin{align}
    &\Var(\bar{x}) = \bar{p}(1 - \bar{p}) ,
\end{align}
where $\bar{p} = \EX_{U \sim \pi_{\mathcal{U}}}  \left[p(U) \right]$, which is exactly the variance of the LCU circuit. The sampling complexity is therefore $O(\frac{1}{\epsilon^2})$ for the renormalized version and $O(\frac{L^2}{\epsilon^2})$ for the non-renormalized one. We observe that due to $p_i^2 \leq p_i$, Eq.~\eqref{eq:no_tot_var_lcu} is always smaller than Eq.~\eqref{eq:var_lcu_ancfree}. 

Let us now consider the same problem from the point of view of amplitude estimation, that is, we implement the SA-LCU estimator given in Eq.~\eqref{eq:SA-LCU} and perform amplitude amplification on it. In this case, the speed up seems not to be present, due to the behaviour of the variance due to the law of total variance \cite{soch2024statproofbook, heinrich2023randomizedbenchmarkingrandomquantum}. We define the total number of queries $n_q$ for AE and the number of random unitary samples $N_U$ needed by this approach to reach a certain precision $\epsilon$. This is not the same as the full LCU, as the sampling of the unitaries from $\pi_{\mathcal{U}}$ is performed classically. The estimator $\tilde{p}_{\text{SA}}$ with $\EX_{U}[\tilde{p}_{\text{SA}}] = \bar{p}$ for the amplified ancilla-free LCU that outputs the amplified value $p_{\text{AE}}(U)$ given a unitary $U \sim \pi_{\mathcal{U}}$ has a variance of:
\begin{align}\label{eq:sa-lcu-totvar}
    &\Var{(\tilde{p}_{\text{SA}})} = \frac{1}{N_U} \Big(\EX_U[p_{\text{AE}}(U)]\left(\frac{1}{n_q^2} - \EX_U[p_{\text{AE}}(U)]\right) + \\  & \hspace{1cm} +\left(1 - \frac{1}{n_q^2}\right)\EX_U[p_{\text{AE}}(U)^2]\Big) \leq \frac{\EX_U[p_{\text{AE}}(U)^2]}{N_U}. \nonumber
\end{align}
For the case $M=1$ we retrieve the variance of the ancilla-free and ancilla-based LCU estimator. It is therefore clear that there is no advantage in increasing the number of samples from the AE estimator, as the reduction in the variance is purely classical and controlled by the number of unitaries sampled from the corresponding distribution rather than the circuit shots. Interestingly, this means that any average of amplified quantum cost functions retains this property. As a consequence, such estimation problems can be tackled, in theory, by using exclusively \textit{single-shot} estimation. \\
On the other hand, the estimator that encodes the sum of unitaries on the quantum circuit can be amplified to the Heisenberg limit. As the ancilla-free LCU can implement the mean value of any observable, we conclude that for any average value of an observable $\EX_{\boldsymbol{\lambda} \sim P} \left[ \langle \mathcal{O}(\lambda) \rangle \right]$ (for example, the average over a dataset in QML), single-shot estimators are a better choice, provided the number of data points available is high enough. This is true for both amplified and non-amplified nested estimation problems \cite{heinrich2023randomizedbenchmarkingrandomquantum}.

\subsection{Example: Quantum machine learning}
Quantum machine learning (QML) refers to the class of algorithms that allow to approximate functions on quantum computers using appropriate families of quantum circuits given a data set. The goal is to perform tasks such as classification, clustering and regression with a certain speedup compared to the classical case  \cite{schatzki2021entangleddatasetsquantummachine, Huang2020}. In this context, we distinguish two main approaches: (1) Algorithms that make use of known quantum routines, such as the HHL algorithm \cite{Harrow2009} and Grover search \cite{Grover1996} and (2) variational algorithms that encode a problem using quantum circuits and variational families of unitaries that are then optimized accordingly. The latter class of algorithms is then commonly sub-divided in explicit and implicit models \cite{Jerbi2023qmlbkm}. In these models, we try to construct an approximator $f_{\btheta}$ with parameters $\btheta \in \mathbb{R}^N$ -- see also Eq.~\eqref{eq:product_gates_ansatz}. Explicit models are defined by a mapping with input $\boldsymbol{x} \in \mathbb{R}^m$:
\begin{align}\label{eq:qml_explicit}
    f({\btheta}, \boldsymbol{x}) = \tr{\rho(\boldsymbol{x}) \mathcal{O}({\btheta})},
\end{align}
with $\mathcal{O}(\btheta) = V(\btheta) \mathcal{O} V(\btheta)^{\dagger}$, while implicit models are given by a linear combination of functions of the data points $\boldsymbol{x}_1, ..., \boldsymbol{x}_M \in \mathbb{R}^m$: $f_{\boldsymbol{\alpha}, \mathcal{D}}(\boldsymbol{x}) = \sum_{m=1}^M \alpha_m k(\boldsymbol{x}, \boldsymbol{x}_m)$
where the kernel is defined as $k(\boldsymbol{x}, \boldsymbol{x}_m) = \tr{\rho(\boldsymbol{x}) \rho(\boldsymbol{x}_m)}$ and $\boldsymbol{\alpha}$ is a vector of parameters. Both quantities can be sampled using the LCU circuit -- see Fig.~\ref{fig:SE_vs_lcu} (b). 
\begin{figure}
    \centering
    \includegraphics[width=0.45\textwidth]{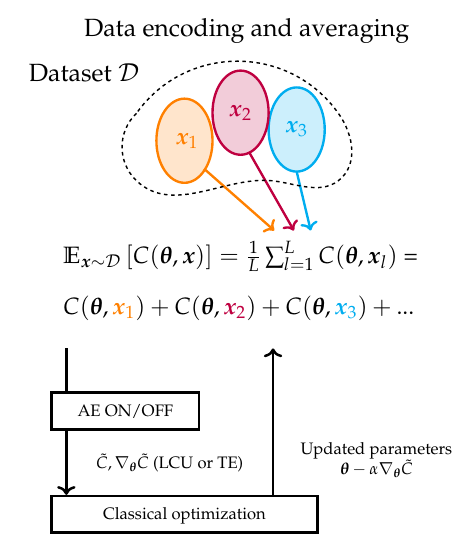}
    \caption{A representation of data encoding and sampling for a QNN \cite{qiskit2024} (explicit model). Classical data needs to be loaded on the quantum sampler/estimator. This procedure can be quite expensive due to the input data size and may require some classical pre-processing \cite{Jerbi2023qmlbkm}, but it can be realized with both LCU and SE approaches. In addition to the two estimators, (near-term) AE routines can be considered \cite{Oshio2024, Brassard2002, Suzuki2020, Grinko2021}. The cost function is controlled by variational parameters $\boldsymbol{\theta}$ that are optimized classically using gradient-based \cite{li2024efficientquantumgradienthigherorder, Motzoi2011, dalgaard2020hessian} or gradient-free \cite{Caneva2011, NelderMead1965} algorithms. Different types of cost functions, such as the one given in Eq.~\eqref{eq:qml_cost}, can be estimated using the LCU method given in Eq.~\eqref{eq:pos_neg_C} or variants thereof.}
    \label{fig:qml_diagram}
\end{figure}
As an illustrative example, we consider the problem of regression for quantum cost functions. When dealing with  (quantum) supervised learning, we are usually given a dataset $\mathcal{D}$ made of pairs of inputs $\boldsymbol{x}_1 \in \mathbb{R}^m, ..., \boldsymbol{x}_L \in \mathbb{R}^m$ and corresponding labels $y_1 \in \left[0,1\right], ..., y_L \in \left[0,1\right]$, i.e., $\mathcal{D} = \{ (\boldsymbol{x}_i, y_i), i=1,...,L\}$. We consider here only two possible label options, i.e., the problem corresponds to binary classification and requires two output qubits. The generalization to the $k$-nary classification problem can be achieved by considering a vector $\boldsymbol{y} \in \{-1,1\}^k$ and using $\log(k)$ qubits, where $k$ is the number of different possible labels in the data. We assume that there exists a function $f: \mathbb{R}^N \times \mathbb{R}^{m} \mapsto [0,1], \left(\btheta ,  \boldsymbol{x}\right) \longmapsto y' = f(\btheta, \boldsymbol{x})$ that depends on parameters $\btheta$ that takes $\boldsymbol{x}_i$ as inputs and outputs guess values $y_i'$. The regression problem is usually solved by minimizing the mean-squared-error loss \cite{Kay1993}, that is by minimizing the following cost function with respect to the model parameters $\btheta$:
\begin{align}\label{eq:qml_cost}
    C_{\text{ML}}(\btheta) = \frac{1}{2L} \sum_{i=1}^L (y_i - f(\btheta, \boldsymbol{x}_i))^2.
\end{align}
The solution $\btheta^*$ to the minimization problem is given by $\btheta^* = \text{argmin}_{\btheta} \ C_{\text{ML}}(\btheta)$ and can be found using different optimization procedures \cite{Ostaszewski2021structure}. Let us now consider the case in which the model is given by a QNN circuit. In this case, the function approximator $f$ is given by a family of parametrized quantum circuits -- see Eq.~\eqref{eq:qml_explicit} --
where the input-dependent density matrix $\rho(\boldsymbol{x})$ is constructed by applying a so-called feature map $V_{\boldsymbol{\phi}}(\boldsymbol{x})$ -- $\boldsymbol{\phi}$ are specific feature encoding parameters -- on the initial quantum state $\rho$ such that $\rho(\boldsymbol{x}) = V_{\boldsymbol{\phi}}(\boldsymbol{x}) \rho V^{\dagger}_{\boldsymbol{\phi}}(\boldsymbol{x})$.
Compared to the standard cost functions of classical machine learning, this cost function is computed as an average over measurement outcomes.
This type of encoding has been used extensively in QML and it is available in Ref.~\cite{qiskit2024}. If we consider an observable $\mathcal{O} = Z_{\text{prod}}$ and measure therefore the parity of the circuit output as a function of the input data, we see that we can simplify the cost function in Eq.~\eqref{eq:qml_cost} using the fact that the squared terms sum to one we are left with the expression \cite{schatzki2021entangleddatasetsquantummachine}:
\begin{align}
    &C_{\text{ML}}(\btheta) = \frac{1}{2} - C_1(\btheta) + C_2(\btheta),\\ 
    &C_1(\btheta) = \frac{1}{L} \sum_{l=1}^L f(\btheta, \boldsymbol{x}_i)y_i\\
    & C_2(\btheta) = \frac{1}{2L} \sum_{l=1}^L f(\btheta, \boldsymbol{x}_i)^2.
\end{align}
It is clear that this is nothing else than a linear combination of estimates that contains both positive and negative coefficients. It can therefore be computed using both of the methodologies outlined in \ref{sec:sampling_sest} (standard approach, SE) and \ref{sec:sampling_lcu} (LCU), respectively. A schematic diagram of a QML sampling process and optimization routine is given in Fig.~\ref{fig:qml_diagram}. \\
The estimation of such quantity represents a possible application of the LCU approach. 
In our case, we make use of the simplified cost function:
\begin{align}\label{eq:qml_lcu_cost}
    C_1(\btheta) = 1 - \frac{1}{L} \sum_{i=1}^L f(\btheta, \boldsymbol{x}_i) y_i.
\end{align}
The second part of the cost function can be sampled by either adding a control operation to the LCU circuit or by sampling with two different circuits. The first is a LCU circuit that implements Eq.~\eqref{eq:qml_lcu_cost}. The second a is a circuit that implements $C_2(\btheta)$. The latter can be realized by preparing the tensor product of the QNN unitary that is used to parameterize $f(\btheta, \boldsymbol{x})$ conditioned on the LCU register. The LCU register implements the summation process and the tensor product allows us to estimate $f(\btheta, \boldsymbol{x})^2$.

\subsubsection{Sampling cost functions from datasets}

We test the implementation of LCU and SE on a QML problem. We employ the Quantum Neural Network (QNN) given in \textsc{qiskit} \cite{qiskit2024}, which is composed of a ZZ feature map for encoding the input data $\boldsymbol{x}$ sampled from the data set and a variational circuit $V(\btheta)$. In this context, we are interested in simply evaluating the average cost function over the $(\boldsymbol{x}_i,y_i)$ input values for $i=1,...,L$ and estimating its variance using the expressions introduced in Eq.~\eqref{eq:SE_estimator} for SE and Eq.~\eqref{eq:var_sz_lcu} for the LCU estimator applied on the cost function in Eq.~\eqref{eq:qml_lcu_cost}. Since the linear combination contains both positive and negative values, we use the expression given in Eq.~\eqref{eq:pnlc_sampling} and modify the variance accordingly. As before, in this context we have the inner expectation value, corresponding to the average over the quantum statistics and the external one, corresponding to the average over the classical sampling. \\
For the plots in Fig.~\ref{fig:qml} we use four different datasets: (I) normally distributed variables $\boldsymbol{x} \sim \mathcal{N}(\boldsymbol{0}_d, \mathbb{I}_d)$, for which we have $y_i = f(\boldsymbol{x}_i) = 2 \cdot \text{sign}(\sum_j x^j_i) - 1$, i.e., the sign of the sum of vector components, -- (II) -- autocorrelated variables, where each variable $\boldsymbol{x}_i$ is defined by the following relation:
\begin{align}\label{eq:autocorr}
    \forall i=2,...,L: \ \boldsymbol{x}_i = i \cdot \boldsymbol{x}_0,\  \boldsymbol{x}_0 \sim \mathcal{N}(\boldsymbol{0}_d, \mathbb{I}_d),
\end{align}
and again $y_i = f(\boldsymbol{x}_i) = 2 \cdot \text{sign}(\sum_j x^j_i) - 1$. The autocorrelated data is used to analyze the behaviour of SE variance in presence of correlations between different circuits. (III) data sampled from the IRIS \cite{online_iris_53} dataset -- (IV) -- data sampled from the MNIST data set \cite{deng2012mnist} after performing a PCA  transformation \cite{Bishop2007} to achieve dimensionality reduction \cite{Jerbi2023qmlbkm}.
\begin{figure*}
    \hspace{-1cm}
    \includegraphics[width=\linewidth]{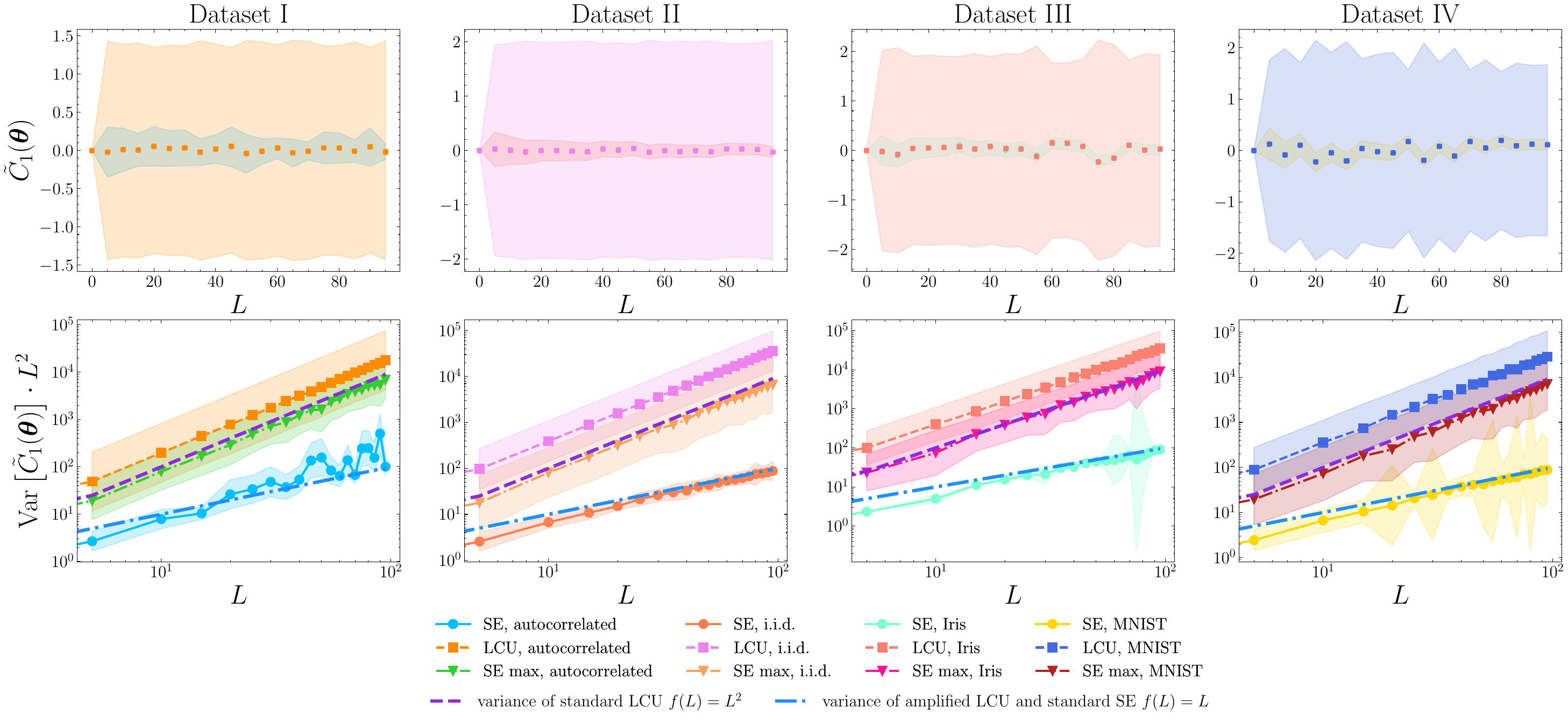}
    \caption{An example of estimation performed with (simulated) quantum circuits in \textsc{qiskit} \cite{qiskit2024} for the regression cost function $C_1(\boldsymbol{\theta})$ in Eq.~\eqref{eq:qml_lcu_cost}, where we use 10000 shots per circuit and from 2 to 100 estimates. The values for datasets I, II and IV are also averaged over 50 different sampling runs. (Top line) Mean values and sampling complexity of the estimator $\tilde{C}_1(\boldsymbol{\theta})$ of $C_1(\boldsymbol{\theta})$ for a random $\boldsymbol{\theta} \sim \mathcal{N}(\boldsymbol{0}_N,\mathbb{I}_N)$ -- see Eq. \eqref{eq:qml_lcu_cost}. The sampling complexity is defined as the asymptotical total number of queries needed to estimate a quantity up to a fixed precision $\epsilon$. (Bottom line) sampling complexities of estimators according to Eqs.~\eqref{eq:var_te} and \eqref{eq:var_sz_lcu} for different types of datasets used as inputs to the QNN for both SE (blue line) and the LCU (orange line) estimators. In addition, the maximum possible SE variance is also shown (green line), and, as expected from Eq.~\eqref{eq:cauchy_schwarz}, it always lies below the LCU variance. Shaded regions show the uncertainty for both mean (standard deviation) and variance (here we use an approximate estimate of the fourth moment for SE and LCU, while for the maximum of SE we use the fourth moment of LCU as an upper bound). Column (I) shows the results of sampling auto-correlated quantities as shown in Eq.~\eqref{eq:autocorr}. (II) shows the results for sampling i.i.d.~Gaussian variables. (III) shows the results of sampling from the IRIS dataset and (IV) from the MNIST dataset (whose dimensionality is reduced first with a PCA transformation, see also the procedure used in Ref.~\cite{Jerbi2023qmlbkm}). We observe that in all cases the variance grows linearly compared to the variance of the LCU, which is always quadratic. In the auto-correlated case, the particular structure of the data seems to induce a superlinear behaviour, most likely due to the fact that estimates that are functions of the same unitaries (or highly correlated unitaries) are considered, see also the discussion in Appendix \ref{sec:sampling_PoissonBinomial}.
    }
    \label{fig:qml}
\end{figure*}
The results of our simulations are shown in Fig.~\ref{fig:qml}. We used 10000 shots for each circuit evaluation, and a number ranging from 2 to 100 of estimates of classical data that are summed together. In these plots we show different types of sampled data: in the first row we plot the values of the cost functions as a function of the number of samples for the four datasets (autocorrelated, uncorrelated, \textsc{IRIS} and \textsc{MNIST}, respectively) for both LCU and SE. The values for datasets I, II and IV are also averaged over 50 different sampling runs. Shaded regions show the standard deviation and the fourth moment of the estimates. We see that the standard deviation for the LCU estimator is always larger than the one of SE. The second row of plots shows the variance of the LCU estimators, SEs and also the maximum possible variance that hypothetic classically correlated (Rademacher) variables $X_1,...,X_L$ with the same mean values as those used for SE. As we showed already, the maximum of such quantity cannot exceed the LCU variance, which is confirmed by the result in each plot. 

Since the cost functions are initially renormalized, in order to obtain the sampling complexity we need to multiply them with an appropriate re-scaling parameter $L^2$ for the LCU estimator and $L^3$ for SE -- because we need to account for the sampling of $L$ different circuits, as shown in Eq.~\eqref{eq:complexitySE}. As a result, we see that the sampling complexities of the two estimators are asymptotically the same.
The maximum possible variance that SE can achieve is due to the Cauchy-Schwarz inequality -- see also Eq.~\eqref{eq:cauchy_schwarz} (assuming classical correlations between the samples are possible). In this case, however, we have to consider that as a result of the correlations, some operations may commute, which would reduce the total number of copies needed to estimate their values. The extreme case is the estimation of the same (Pauli) observable multiple times: in this case the LCU and SE become the same estimator. Interestingly, this is also the (classical) case in which the so-called Poisson-Binomial probability distribution \cite{Wang1993} and its Binomial approximation are the same -- see the discussion in Appendix \ref{sec:sampling_PoissonBinomial}. This aspect helps us shed light on the different properties of LCU and SE and how these differences affect the scaling of estimation on quantum circuits. It also shows that full LCU is mostly indicated for sampling and estimation tasks where full or at least near-term amplitude amplification algorithms can be implemented.

\begin{figure*}[ht!]
    \hspace{-1cm}
    \includegraphics[width=18.5 cm]{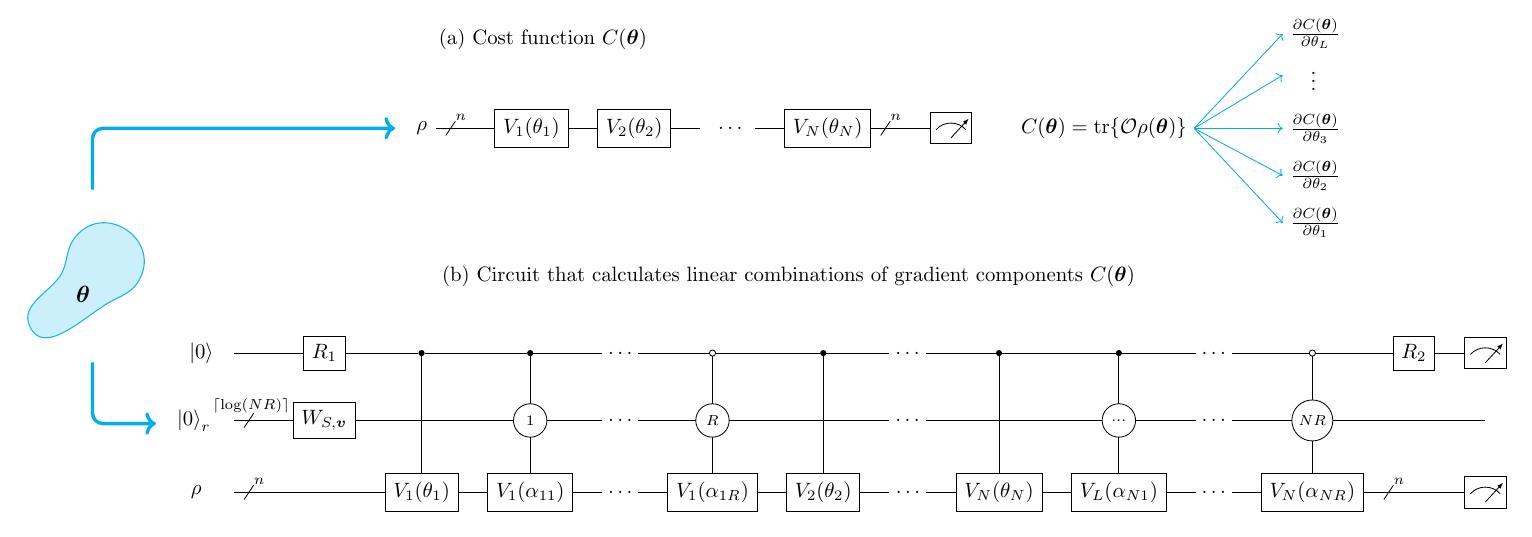}
    \caption{Schematic representation of the forward derivative circuit for a general quantum cost function given in Eq.~\eqref{eq:h_v_sigma_z}. The cost function circuit is pictured in (a). In the case of parameter-shift rules, the cost function is evaluated multiple times by shifting the parameter vector along the derivative axis for instance according to Eq.~\eqref{eq:psr}, and where $\alpha_{ir},\ i=1,..., N, r=1,...,R$ represent the shifts in the respective angles. The gradient circuit \cite{Somma2002} (b), instead, uses a LCU register (or two registers for positive-negative coefficients) with controlled operations to perform either parameter-shift rules or finite-differences \cite{Wierichs2022, LiJun2017, Izmaylov2021, bittel2022fast, Mari2021}, or to encode the gradient operator itself \cite{Schuld2019, li2024efficientquantumgradienthigherorder}. The gradient circuit uses multi-controlled gates on the values $1,...,NR$ (assuming w.l.o.g.~that each gate $V_1,...,V_N$ requires a total of $R$ shifts). Controlled operations on the values $1,..., NR$ can be implemented using binary encoding and $\lceil \log(NR) \rceil$ qubits. The LCU unitary $W_{S,\boldsymbol{v}}$ encodes both the coefficient matrix $S$ in Eq.~\eqref{eq:psr} and the forward derivative coefficients $\boldsymbol{v}$ in Eq.~\eqref{eq:for_dev}. To obtain the standard gradient of the quantum cost function the circuit output needs to be evaluated $N$ times with $N$ different binary input vectors. The forward derivative, instead, can be evaluated directly by encoding a vector $\boldsymbol{v}$ in the LCU register with $r=\lceil \log(NR) \rceil$ qubits (we assume here for the sake simplicity only positive linear combinations). The forward gradient can be estimated by using a random input $\boldsymbol{v} \sim Q$, where $Q$ is a suitable probability distribution, see Ref.~\cite{Baydin2018}.
    }
    \label{fig:forward_dev_circ}
\end{figure*}
\section{Gradients of quantum cost functions}\label{sec:grad_cost_functions}
In this Section, we consider the problem of sampling and estimating gradients of quantum cost functions. This problem has been discussed in several previous works \cite{Schuld2020, Schuld2019, Wierichs2022, Wiersema2024herecomessun, LiJun2017, Banchi2021measuringanalytic, crooks2019gradientsparameterizedquantumgates, Izmaylov2021, Kyriienko2021, Mitarai2018, Wierichs2022, preti2024hybrid, li2024efficientquantumgradienthigherorder}, both from the point of view of LCU-type approaches and SE. Here, we want to focus in particular on the circuit shown in Fig.~\ref{fig:forward_dev_circ} (b) and apply it to the following problems: the estimation of derivatives with forward propagation \cite{Baydin2018}, see Section \ref{sec:forward_prop}, its application in the calculation of general SU($d$) gradients \cite{Wiersema2024herecomessun}, see Section \ref{sec:sun_gradient_lcu} and Fig.~\ref{fig:sun_gradient_circuit}, and its application in quantum control \cite{Khaneja2005}, see Section \ref{sec:q_control_gradients}. These altogether complete the description of gradient estimation circuits using LCU approaches and allow for the estimation of arbitrary gradients of unitary operations.
\begin{figure}[ht!]
    \centering
    \includegraphics[width=0.5\textwidth]{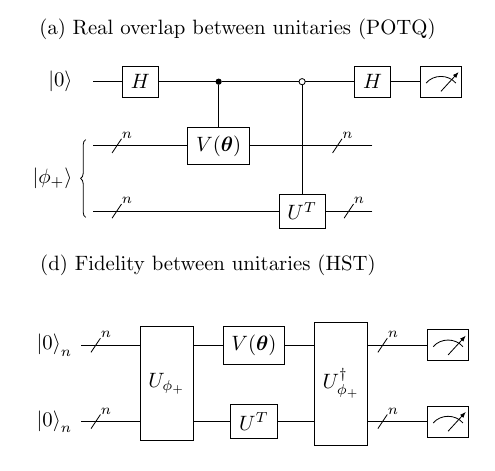}
    \caption{Relevant test circuits for quantum compilation introduced in Ref.~\cite{Khatri2019}: (a) Power-of-Two-Qubits (POTQ) circuit (the original article uses one qubit more, but the result is analogous) and fidelity or Hilbert-Schmidt test circuit (b). The operation $U_{\phi^{+}}$ \cite{Khatri2019} prepares the $2n$-qubit generalized Bell state $\ket{\phi^{+}}$ given in Eq.~\eqref{eq:bell_gen} starting from $\ket{0}^{\otimes 2n}$. See Ref.~\cite{li2024efficientquantumgradienthigherorder} for a general review of gradients based on Hadamard tests in the context of parametrized circuits.}
    \label{fig:test_circ_lcu}
\end{figure}
\subsection{Quantum circuit gradients}
For general gradients of quantum cost functions -- see Eq.~\eqref{eq:SEst} -- we can write:
\begin{align}\label{eq:cost_function_grad}
    \nabla_{\btheta} \langle \mathcal{O}(\boldsymbol{\theta}) \rangle = \Re{\tr(\nabla_{\btheta} V(\btheta) \rho V^{\dagger}\mathcal{O})},
\end{align}
which, as shown in Fig.~\ref{fig:test_circ_lcu}, can be again implemented by combining LCU and a Hadamard test \cite{li2024efficientquantumgradienthigherorder}.
Let us consider the outcome of the so-called Power-of-Two-Qubits circuit or POTQ \cite{Khatri2019} circuit, a specific type of Hadamard test that outputs the following overlap between unitaries $V(\btheta)$ and $U^T$ and which is shown in Fig.~\ref{fig:test_circ_lcu} (a):
\begin{align}\label{eq:potq_cost}
    p_{\pm}(\btheta) = \frac{1}{2}\left(1 \pm \frac{1}{d}\Re{\tr{V(\btheta) U^{\dagger}}}\right),
\end{align}
where $p_{\pm}(\btheta)$ are the probability of finding the control qubit in the state 1 or 0, respectively. The cost function corresponding to the real trace overlap can then be estimated as $C_{\text{POTQ}}(\btheta) = 1 - p_{+}(\btheta) - p_{-}(\btheta) = 1-\frac{1}{d}\Re{\tr{V(\btheta) U^{\dagger}}}$. The gradient of this cost function is given by:
\begin{align}\label{eq:potq_grad}
    \nabla_{\btheta} C_{\text{POTQ}}(\btheta) = -\frac{1}{d}\Re{\tr{\nabla_{\btheta} V(\btheta) U^{\dagger}}}.
\end{align}
Eq.~\eqref{eq:potq_grad} can be estimated using again a POTQ circuit where the unitary gradient $\nabla_{\btheta} V(\btheta)$ is implemented. 

We first consider the specific case in which the unitary $V(\btheta)$ has a single parameter $\theta$, i.e., $V(\theta) = \exp{-i \theta H}$ for a Hamiltonian $H$. Assuming the Hamiltonian can be written as $H = \sum_{i=1} h_i P_i$ \cite{Schuld2019, Schuld2020}, this gradient can be implemented using the circuit given in Fig.~\ref{fig:SE_vs_lcu} (b), since $\dot V(\theta)=-iHV(\theta)$. For a quadratic cost function, such as the standard gate infidelity -- see Fig.~\eqref{fig:test_circ_lcu} and Ref.~\cite{Khatri2019}:
\begin{align}\label{eq:infidelity}
    I(\btheta) = 1 - \frac{1}{d^2} \abs{\tr{V(\btheta)U^{\dagger}}}^2,
\end{align}
we have instead a gradient of the type:
\begin{align}\label{eq:fidelity_grad}
    \nabla_{\btheta} I(\btheta) = - \frac{2}{d^2} \Re{\tr{\left( V(\btheta)U^{\dagger} \right) \otimes \left( \nabla_{\btheta}V(\btheta)U^{\dagger} \right) }},
\end{align}
which can be sampled by modifying the POTQ circuit with a register of $4n$ qubits and encoding the gradient of the unitary operation on the circuit appropriately \cite{li2024efficientquantumgradienthigherorder}.

\subsection{Parameter-shift rules}
\vspace{-1cm}
Let us now consider the circuit product ansatz:
\begin{align}\label{eq:product_gates_ansatz}
    V_{\text{CA}}(\btheta) = \prod_{i=1}^N V_i(\theta_i),
\end{align}
with $\btheta \in \mathbb{R}^N$. In simulation, the gradient of this circuit can be effectively computed using the GRAPE algorithm \cite{Khaneja2005}. However, its evaluation on a real quantum device seems challenging, as it is not possible to naively store the results of intermediate unitaries on quantum experiments as we do in classical simulations, though some proposals for back-propagation on quantum circuits do exist \cite{bowles2024backpropagationscalingparameterisedquantum, abbas2023quantumbackpropagationinformationreuse}. 
In the case where Eq.~\eqref{eq:product_gates_ansatz} is valid, we see that the gradient of the cost function given in Eq.~\eqref{eq:cost_function_grad} simplifies to:
\begin{align}
    \pdv{}{\theta_i} \langle \mathcal{O}(\boldsymbol{\theta}) \rangle = i{\Tr{\mathcal{O} V_{(1)}^{(i)}(\btheta) [H_i, \tilde{\rho}_i(\boldsymbol \theta)] V_{(1)}^{(i)\dagger}(\btheta)}},
\end{align}
where $\tilde{\rho}_i(\boldsymbol \theta) = V_{(i+1)}^{(N)}(\btheta) \rho V_{(i+1)}^{(N)\dagger}(\btheta)$ and $V_k^j(\boldsymbol \theta) = \prod_{l=N}^L V_l(\theta_l)$.
By finding appropriate coefficients $S^i_k$ and parameter shifts $\alpha_{ik}$, we can use the interpolation ansatz \cite{Wierichs2022}:
\begin{align}\label{eq:psr}
    \pdv{}{\theta_i} \langle \mathcal{O}(\boldsymbol{\theta}) \rangle= \sum_{k=1}^R S_{ik} \langle  \mathcal{O}(\boldsymbol{\theta} + \alpha_{ik}\boldsymbol{e}_i) \rangle,
\end{align}
where $R$ is the number of parameter shifts (in Ref.~\cite{Wierichs2022} it is set as twice the number of distinct eigenvalue differences contained in the spectrum of the gate generator) and $\boldsymbol e_i,\ i=1,...,N$ are unit vectors of the parameter space $\mathbb{R}^N$ -- see Eq.~\eqref{eq:product_gates_ansatz}. We can substitute symbolically any type of product ansatz into Eq.~\eqref{eq:psr} to find a suitable formula for the parameter-shift vectors $\boldsymbol{\alpha}_1=(\alpha_{11}, ...,\alpha_{1R})^{\text{T}}, ..., (\alpha_{N1}, ...,\alpha_{NR})^{\text{T}}$ and the coefficient matrix $S_{ik}$. This expression is the general parameter-shift rule for quantum gradients \cite{Wierichs2022}, written using a different notation.
Moreover, even in the noiseless case, one needs $O(R^2N/\epsilon^2)$ circuit evaluations to compute every parameter variation. The same is true if the gradient is computed through, e.g., finite-difference methods \cite{Mari2021}, with the notable exception that the scaling for one finite-difference shift is $O\left[N/(\epsilon^2 \delta^2)\right]$ \cite{Kyriienko2021}, which is generally unfavorable due to the small size of $\delta$. In this case, a shift $\delta \sim O(1)$ is often the optimal choice \cite{Mari2021, Wiersema2024herecomessun}.

It is clear that Eq.~\eqref{eq:psr} can also be implemented using the LCU sampler -- see Eq.~\eqref{eq:pnlc_sampling}. Pauli gates are a simple application, as they only have two distinct eigenvalues \cite{LiJun2017}. The $n$-qubit XY gate or M{\o}lmer-S{\o}rensen gate offer a more interesting use case with $n + 1$ and $\lceil n/2 \rceil + 1 $ distinct eigenvalues \cite{preti2024hybrid, consiglio2025variationalquantumalgorithmsmanybody} respectively, which correspond to $O[\log(n)]$ control qubits in the LCU register. However, for arbitrary generators the eigenvalue structure may also scale more unfavourably.

For arbitrary generators with partially known spectra, several methods of reconstructing their landscape have been proposed. Often, these methods require performing linear transformations of a vector of estimates, such as computing their Discrete Fourier Transform (DFT). In this case, however, the coefficients need to be determined from the estimates classically, which means that loading them on the LCU register is less beneficial than approaching the problem using the method given in Ref.~\cite{Schuld2019}. 

\subsection{Directional derivatives and forward propagation}\label{sec:forward_prop}

Another application of LCU is represented by forward propagation \cite{Baydin2018}. While back-propagation seems to be challenging to implement on quantum computers \cite{bowles2024backpropagationscalingparameterisedquantum, abbas2023quantumbackpropagationinformationreuse}, forward propagation on quantum circuits can instead be achieved by using the circuit shown in Fig.~\ref{fig:grape_circuit} (d) or by implementing it with the SE. The forward derivative of a cost function $C$ that depends on parameters $\btheta \in \mathbb{R}^N$ is defined as:
\begin{align}\label{eq:for_dev}
    \nabla_{\boldsymbol{v}} C(\btheta) = \sum_{i=1}^N v_i \pdv{C(\btheta)}{\theta_i},
\end{align}
for a vector $\boldsymbol{v} \in \mathbb{R}^N$.
And the corresponding forward gradient can be obtained by sampling in random directions \cite{Baydin2022gwb, belouze2022optimization}:
\begin{align}
    \nabla_{\btheta} C(\btheta) = \EX_{\boldsymbol{v} \sim Q} \left[ \nabla_{\boldsymbol{v}} C(\btheta) \boldsymbol{v} \right],
\end{align}
where $\boldsymbol{v} \sim Q$ is a random vector with mean zero and bounded variance sampled from a probability distribution $Q$. In Ref.~\cite{belouze2022optimization} $Q$ is assumed to be a Rademacher distribution: values sampled from this distribution can be implemented in the LCU register using random graph states \cite{Wu_2014_random}. Since $\boldsymbol{v}$ can have both positive and negative values, the circuits given in Fig.~\ref{fig:positive_negative_lcu} are needed. 
Despite its straightforward circuit-centric implementation, the scaling of forward propagation on quantum computers is significantly worse than the corresponding scaling on classical deterministic computers \cite{Baydin2018}. We consider here the product ansatz circuit given in Eq.~\eqref{eq:product_gates_ansatz}. We assume that each gate can be written as $V_i(\theta_i) = \exp{- i \theta_i H_i} = \exp{- i \theta_i \sum_{k=1}^L a_{ik}P_{ik}}$, where $P_{ik}$ are arbitrary Pauli operators and $a_{ik}$ the coefficients of the Hamiltonian of the $i$th gate in the Pauli basis \cite{Schuld2019}. The relevant scaling parameters of forward propagation are the number $N$ of gates in the product ansatz and the number of Pauli coefficients $K$ in each gate according the LCU approach -- see Eq.~\eqref{eq:H}. Therefore, the LCU circuit can estimate:
\begin{align}\label{eq:forw_pauli}
    &g_{\boldsymbol{v}} = -2\sum_{i=1}^N v_i \sum_{k=1}^K a_{ik} \Im{\tr{P_{ik} \rho_{CA}(\btheta) \mathcal{O}}},\\
    &\hspace{1.8cm}\nabla_{\boldsymbol{\theta}}C(\btheta) = \EX_{v \sim Q}[g_{\boldsymbol{v}}],
\end{align}
for $\rho_{CA}(\btheta) = V_{CA}(\btheta) \rho V^{\dagger}_{CA}(\btheta)$. Eq.~\eqref{eq:forw_pauli} is the forward derivative in Eq.~\eqref{eq:for_dev} of the cost function in Eq.~\eqref{eq:SEst} that uses the circuit defined in Eq.~\eqref{eq:product_gates_ansatz}. 
Due to propagation of uncertainty, we obtain a $O(\frac{N^2K^2}{\epsilon^2})$ scaling for the forward gradient, which is worse than the PSR and LCU-gradient scaling with respect to $N$. Eq.~\eqref{eq:forw_pauli} can be estimated by combining appropriate parameter shifts and forward propagation -- see Fig.~\eqref{fig:forward_dev_circ} (b) --, i.e., with a sampling complexity of $O(\frac{N^2R^2}{\epsilon^2})$. An amplified version of forward propagation computed on the LCU circuit would reduce the sampling complexity of forward propagation to $O(\frac{NK}{\epsilon})$. In this case, we reach the same scaling as the PSR approach in $N$ and quadratically better scaling in $L$. This analysis, however, does not cover the total variance of the estimation problem due to the number of shots and the variance with respect to $\boldsymbol{v} \sim Q$. In fact, forward propagation suffers from slowdowns in the optimization due to a curse of dimensionality that limits its effectivity even in classical settings \cite{belouze2022optimization}. 
\subsection{Most general LCU gradient: SU($d$) gradient circuit}\label{sec:sun_gradient_lcu}
In this case the circuit is a single multi-qubit circuit parametrized by: 
\begin{align}\label{eq:sun_gate}
    V(\btheta) = e^{-i H(\btheta)} = \exp{-i \sum_{k=0}^{K-1} \theta_k H_k}.
\end{align} 
The circuit that provides an unbiased estimator for the adjoint is given in Ref.~\cite{li2024efficientquantumgradienthigherorder} and Fig.~\ref{fig:sun_gradient_circuit}. Now we need to combine this circuit with the LCU itself. In the case of a general SU($d$) gate, Eq.~\eqref{eq:potq_grad} becomes \cite{Wiersema2024herecomessun}:
\begin{align}\label{eq:potq_gradient}
    \nabla_{\btheta} C_{\text{POTQ}}(\btheta) = -\frac{1}{d}\Re{\Tr{ \Omega(\btheta) V(\btheta) U^{\dagger}}},
\end{align}
with the operator
\begin{align}\label{eq:SUN_gradient_omega}
    \Omega(\btheta) = \sum_{l=0}^{\infty} \frac{(-i)^l}{(l+1)!} \norm{\btheta}_1^l \ad^{l}_{\bar{H}(\btheta)} \left(\bar{\nabla}_{\btheta} H(\btheta) \right),
\end{align}
where $\ad^{l}_{\bar{H}(\btheta)} \left(\bar{\nabla}_{\btheta} H(\btheta) \right)$ denotes the nested commutator of $l$th order
\begin{align}
    \ad^{l}_{\bar{H}(\btheta)} \left(\bar{\nabla_{\btheta}} H(\btheta) \right) = \underbrace{\left[\bar{H}(\btheta), ..., \left[\bar{H}(\btheta), \left(\bar{\nabla}_{\btheta} H(\btheta) \right) \right] \right]}_{l+1 \ \text{elements}}, \nonumber
\end{align}
and $\bar{H}$ is the renormalized version of the Hamiltonian, i.e., $\bar{H}(\btheta) = H(\btheta)/\norm{\btheta}_1$ with $\norm{\btheta}_1 = \sum_{k=0}^{K-1} \abs{\theta_k}$, as defined in Eq.~\eqref{eq:sun_gate}. The expression $\bar{\nabla}_{\btheta} = \frac{1}{\norm{\btheta}_1} \nabla_{\boldsymbol{\theta}}$ is a rescaled nabla operator. As such, the derivative of the unitary for $k=0,...,K-1$ given in Eq.~\eqref{eq:sun_gate} can be written as an infinite sum of matrices:
\begin{align}
    \pdv{}{\theta_k}V(\boldsymbol{\theta}) = \sum_{l=0}^{\infty} \mathcal{W}_l(\boldsymbol{\theta}),
\end{align}
where $\mathcal{W}_l(\boldsymbol{\theta}) = \frac{(-i)^l}{(l+1)!} \norm{\btheta}_1^l \ad^{l}_{\bar{H}(\btheta)} \left( H_k \right) V(\btheta)$. If we combine the estimator circuit given in Eq.~\eqref{eq:pnlc_sampling} with the circuit given in Ref.~\cite{li2024efficientquantumgradienthigherorder}, which allows for the computation of nested commutators -- see Fig.~\ref{fig:sun_gradient_circuit} -- we can compute any unitary gradient with a full quantum circuit up to a desired approximation order. A second LCU register is necessary to encode $\bar{H}(\btheta)$ coherently on the quantum circuit.
By inserting the expansion given in Eq.~\eqref{eq:SUN_gradient_omega} into our gradient Eq.~\eqref{eq:potq_gradient}, we obtain 
\begin{equation}
\nabla_{\btheta}C_{\text{POTQ}}(\btheta) = \sum_{l=0}^{\infty} T_l,
\end{equation}
with 
\begin{equation}
T_l = \frac{(-1)}{(l+1)!} \|\theta\|^l \frac{1}{d} \Re{ (-i)^l \Tr\bigl[V \ad_H^l(G) U^\dagger\bigr]}, \label{eq:single_order}
\end{equation}
where $G = \frac{1}{\norm{\btheta}_1} \nabla_{\btheta} H(\btheta) = (H_0, ..., H_{K-1})^{\text{T}}$ is a vector of matrices and the trace and matrix multiplications operations in Eq.~\eqref{eq:single_order} are carried out in a vectorized form.
For an $L$th order expansion 
\begin{equation}\label{eq:sun_lthorder}
\nabla_{\btheta}C_{\text{POTQ}}(\btheta) \approx \sum_{l=0}^{L} T_l,
\end{equation}
the remainder is then given by $R_L=\sum_{l=L+1}^\infty T_l$. Here we focus specifically on the POTQ cost function. However, similar expressions can be derived for arbitrary cost functions of the type given in Eq.~\eqref{eq:SEst}, whose gradient is given by Eq.~\eqref{eq:cost_function_grad}, as well as for the infidelity gradient given in Eq.~\eqref{eq:fidelity_grad}. We now turn to specific examples of multi-parameter, multi-qubit gates where Eq.~\eqref{eq:sun_lthorder} can be implemented using LCU methods. In particular, we analyze the behaviour of the SU($d$) gradient approximation given in Eq.~\eqref{eq:sun_lthorder} on control problems with fixed and time-dependent control parameters.

\begin{figure}
    \centering
    \includegraphics[width=0.5\textwidth]{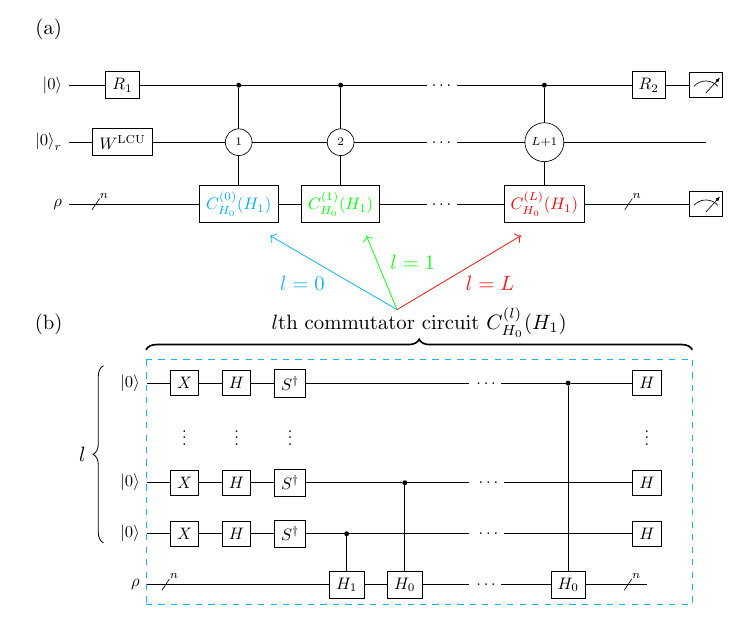}
    \caption{Representation of the SU($d$) gradient circuit \cite{Wiersema2024herecomessun} using the nested commutator circuit from Ref.~\cite{li2024efficientquantumgradienthigherorder} for a control problem with Pauli gates $H_0$ and $H_1$. The circuit $C^{(l)}_{H_0}(H_1)$ shown in (b) computes the expected value of the adjoint $\ad^l_{H_0}(H_1)$, i.e., the nested commutator term of order $l$. Each commutator term of order $l=0$ (which is simply $H_1$), $l=1$ (which is $[H_0,H_1]$), etc.~is mapped to the corresponding entry in the LCU register with $r=\lceil \log(L) \rceil$ qubits. The LCU state is prepared by the unitary $W^{\text{LCU}}$. The circuit in (a) can be simplified further due to the redundancies in the composition of $L$ multiply-controlled circuits of type (b), each one with a number of Hamiltonian terms that grows from $1$ to $L$ as in Eq.~\eqref{eq:sun_lthorder}.}
    \label{fig:sun_gradient_circuit}
\end{figure}

\subsubsection{Example}
We consider now a simplified unitary operation that is typical of control problems \cite{Machnes2011}, i.e.:
\begin{align}
    &V_c(\btheta) = \exp{-i(\theta_0 H_0 + \theta_1 H_1)\Delta t} \label{eq:simple_control_problem} \\
    &V_0 =  \exp{-i \theta_0 H_0 \Delta t}, \label{eq:simple_control_problem_V0}
\end{align}
where $H_0$ and $H_1$ are drift and control Hamiltonians respectively, $\theta_1$ is  a control parameter and $\Delta t$ is a time interval. This is a special case of Eq.~\eqref{eq:sun_gate}, where $\theta_k=0$ for $k=1,...,K-1$ and an additional dynamical evolution of time $\Delta t$ is applied. We define therefore the vector parameter $\btheta = (\theta_0, \theta_1)^T$. We can compute the control gradient of Eq.~\eqref{eq:simple_control_problem} by evaluating $\pdv{V_c(\btheta)}{\theta_1} \Big |_{\btheta = (1, 0)^T}$ using Eq.~\eqref{eq:SUN_gradient_omega}:
\begin{align}\label{eq:h0h1_grad}
    \pdv{V_c(\btheta)}{\theta_1} \Big |_{\btheta = (1, 0)^T} = \sum_{l=0}^{\infty} \mathcal{W}^c_l,
\end{align}
with $\mathcal{W}^c_l = \frac{(-i \Delta t)^l}{(l+1)!} \ad^{l}_{H_0} \left(H_1 \right) V_0$. The circuit given in Fig.~\ref{fig:sun_gradient_circuit} can be used to compute the derivative by implementing $H_0$ and $H_1$ using the LCU method if the Hamiltonians are given by sums of Pauli operators \cite{Schuld2019, Childs2012}. 
\begin{figure*}[ht!]
    \hspace{-0.75cm}
    \includegraphics[width=18 cm]{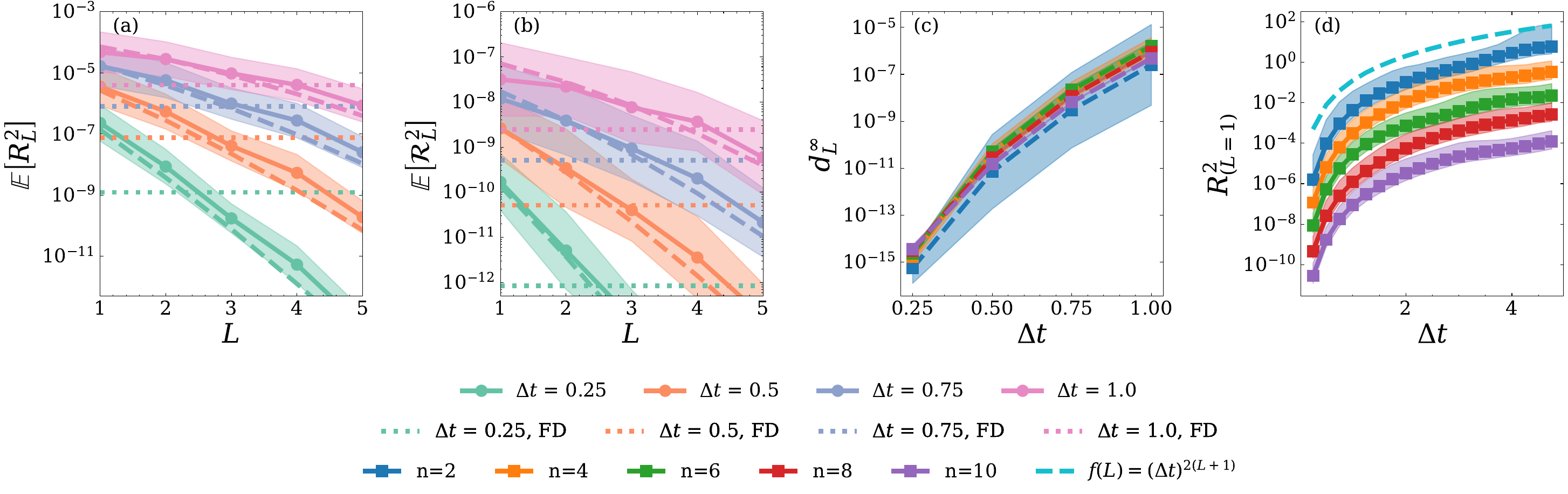}
    \caption{Representation of the convergence of the SU($d$) derivative given in Eq.~\eqref{eq:sun_lthorder}, using the parameters defined in Eq.~\eqref{eq:h0h1_grad} and as such represented using LCU. We use here the squared approximation error $R_L^2$ -- POTQ, Eq.~\eqref{eq:potq_grad} -- and $\mathcal{R}_L^2$ -- infidelity, Eq.~\eqref{eq:fidelity_grad} --, so that we can compare it easily with our theoretical predictions. Numerical simulations of the truncation error are shown as straight lines and the theoretical prediction as dashed lines, with different colours representing different $\Delta t$. The values are averaged over 50 random Hamiltonians $H_0$ and $H_1$ sampled from the GUE \cite{Tao2012-randommatrix} with Dyson coefficient $\Lambda_0 = \Lambda_1 = 1$. (a), (b) show the gradient approximation of the unitary of a $6$-qubit Hamiltonian for the POTQ cost function and the fidelity cost function, respectively, and for different values of $\theta$ as a function of the order of approximation, $n$, in the expansion $L$. Dotted lines represent the symmetric $O(1)$ finite-difference (FD) approach with a shift of $\delta=0.75$. The $O(1)$ approach is the best choice when sampling classically from quantum circuits due to the scaling of the variance with $\delta$ -- $O(1/\delta^2)$, see also Ref.~\cite{Wiersema2024herecomessun}. (c) shows the maximum norm of the difference between the true gradient computed with JAX and the gradient approximation as a function of $\theta$ -- see Eq.~\eqref{eq:infimum_norm}. (d) shows the precision in the approximation of the first order in the expansion for the POTQ cost function, which for this case corresponds to terms in the LCU that are proportional to $H_1$ ($L=0$) and  $[H_0, H_1]$ ($L=1$). We see that for several cost functions, a relatively short circuit already produces a high-quality approximation of the SU($d$) derivative.}
    \label{fig:sun_gradient_sim}
\end{figure*}
\subsubsection{LCU implementation}
We first consider only two Pauli operators, i.e., $H_0=P_0$ and $H_1=P_1$. This helps us illustrate the structure of the quantum circuit that we generalize later to arbitrary Hamiltonians using standard LCU methods. In this case, the circuits given in Fig.~\ref{fig:sun_gradient_circuit} can be implemented using (multi-controlled) single-qubit rotations $R_x, R_y$ and $R_z$ and (multi-qubit) Toffoli gates. 

The number of multi-controlled gates needed to implement a gradient approximation circuit of order $L$ scales as $O[L^2 \log(L)]$. For general Hamiltonians $H_0$ and $H_1$ generated by $K$ Pauli operators, the overall depth is $O[K^2 L^2 \log(K) \log(L)]$. Note as before that the factor of $\log(L)$ can be further brought down to $\log(\log(L))$ by using $\log(L)$ additional ancillas \cite{Motzoi2017}. We observe further that we do not need all multi-controlled bits for the (multi-controlled) Pauli gates that implement the controlled operations $P_1, P_0, ..., P_0$ : the first one, $P_1$, acts on all entries of the LCU state and does not need controls on the LCU register, but only on the first qubit. The other gates are multi-controlled on $L-1$, $L-2$,..., $1$ entries, respectively, so they require only $1$, $1$, $2$, ..., $L$ multi-controlled bits each. This approach reduces the number of gates by a constant factor, even though the asymptotic circuit depth remains the same. 

We now turn to the implementation of the LCU register. Each commutator circuit ($l=0,...,L$) in Fig.~\ref{fig:sun_gradient_circuit} (a) allows for the estimation of the (renormalized) mean value -- see Eq.~(15) in Ref.~\cite{li2024efficientquantumgradienthigherorder}:
\begin{align}\label{eq:G_sun_circuit}
   g_l = \left(-\frac{1}{2} \right)^l \Re{(-i)^l \tr{\ad^{l}_{H_1}(H_0) \rho_{V_0} \mathcal{O}}},
\end{align}
where $\rho_{V_0} = V_0 \rho V_0^{\dagger}$ -- see Eq.~\eqref{eq:simple_control_problem_V0}. As a consequence, the LCU circuit in Fig.~\ref{fig:sun_gradient_circuit} (a) estimates the quantity:
\begin{align}
    \mathcal{G}_L = \sum_{l=0}^L w_l(\Delta t) g_l,
\end{align}
where $w_l(\Delta t)$ are appropriate $\Delta t$-dependent LCU weights that reproduce the coefficients given in Eq.~\eqref{eq:h0h1_grad}. By comparing Eq.~\eqref{eq:G_sun_circuit} with Eq.~\eqref{eq:h0h1_grad}, we see that the state we need to prepare in the LCU register is nothing else than a coherent state with parameter equal to $\sqrt{2\Delta t}$ implemented by the operator $W^{\text{LCU}}$ shown in Fig.~\ref{fig:sun_gradient_circuit} (a):
\begin{align}\label{eq:lcu_coherent}
     &W^{\text{LCU}}(\Delta t) \ket{0}_L = \ket{\sqrt{2\Delta t}}, \\
     &\ket{\sqrt{2\Delta t}} = e^{-\abs{\Delta t}} \sum_{l=0}^{\infty} \frac{(\sqrt{2\Delta t})^{l}}{\sqrt{(l)!}} \ket{l}.
\end{align} 
The positive weights are given by $w_l(\Delta t) = \frac{(2 \Delta t)^l}{l!}e^{-2\abs{\Delta t}}$.
We observe that compared to Eqs.~\eqref{eq:SUN_gradient_omega} and \eqref{eq:h0h1_grad}, using this state would lead to a $l!$ coefficient rather than $(l+1)!$. The derivation of this gradient encoding scheme is provided in Appendix \ref{sec:alpha_gradient}.
A generalization to the case in which the gradient is not computed at $\theta_1 = 0$ and the Hamiltonians are not just $n$-qubit Pauli operators (for example because there are multiple control Hamiltonians that are kept fixed while performing the derivative with respect to $\theta_1$) can be obtained using the LCU approach to implement the drift Hamiltonian and by encoding coefficients proportional to $(\Delta t)^l$ in the LCU register, see Eq.~\eqref{eq:lcu_coherent}. 

\subsubsection{Expectation value of the truncation error}
We want to analyze the convergence behaviour of the gradient approximation in Eq.~\eqref{eq:sun_lthorder} for a quantum gate of the type given in Eq.~\eqref{eq:simple_control_problem}. In particular, we consider here the gradient of the POTQ test given in Eq.~\eqref{eq:potq_grad}. The square truncation error $R_L$ is given by:
\begin{align}
    R_L^2 = \left(\sum_{l=L+1}^{\infty} T_l\right)^2,
\end{align}
where $T_l$ is given in Eq.~\eqref{eq:single_order}. Here we use $R_L$ to describe the square truncation error of the POTQ gradient, while $\mathcal{R}_L$ is defined analogously for the gradient of the infidelity. We can estimate the behavior of the expected value of the truncation error as a function of a random Hamiltonian generator, in order to potentially determine the truncation length $L$ that is needed to achieve a given target precision in the gradient approximation. In the following section, we specifically consider the framework of quantum control to understand the behaviour of the SU($d$) gradient from the perspective of the LCU implementation. We employ random $n$-qubit drift and control Hamiltonians $H_0$ and $H_1$ that are sampled from the Gaussian Unitary Ensemble (GUE) \cite{Tao2012-randommatrix}, respectively. We also employ random target unitaries $U$, generated with a QR decomposition. The gradient is evaluated at control value equal to zero ($\theta_1 = 0$) to facilitate the simulation, but this method can be applied to arbitrary control problems, since we can always redefine the drift Hamiltonian to contain control Hamiltonians, whose pulse parameters are not varied.
\begin{figure*}[ht!]
    \hspace{-1.2cm}
    \includegraphics[width=18cm]{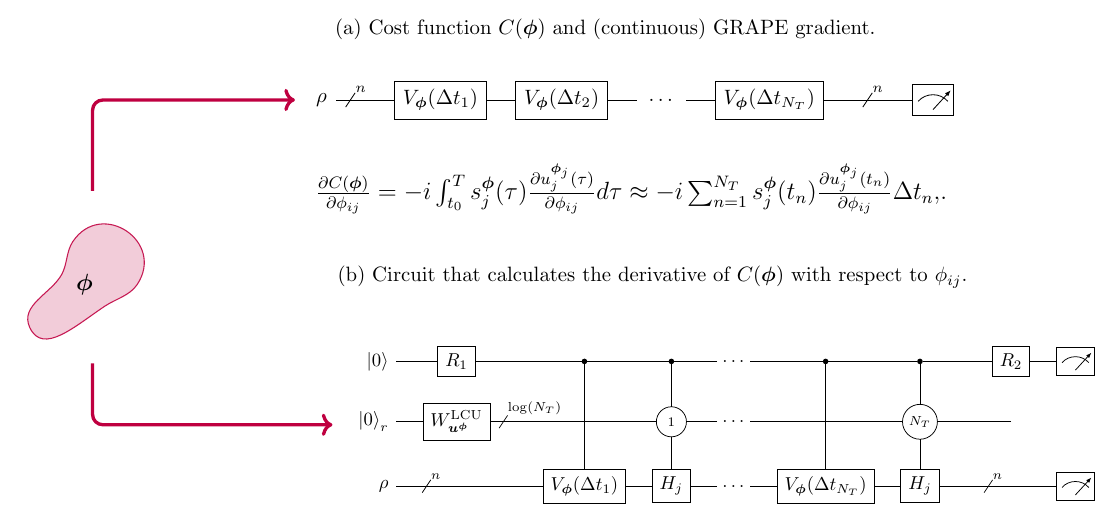}
    \caption{Schematic representation of the LCU-GRAPE circuit for a general quantum cost function optimized with control pulses. The cost function circuit is pictured in (a). Compared to the previous example, i.e., Fig.~\ref{fig:forward_dev_circ}, the circuit uses the unitary evolution $V_{\boldsymbol{\phi}}(t_0,T) = \mathcal{T}\exp{- i \int_{t_0}^T H_{\boldsymbol{\phi}}(\tau) d\tau}$, where $\mathcal{T}$ is the time-ordering operator and $H(t) = H_0 + \sum_{j=1}^{K-1} u^{\boldsymbol{\phi_j}}_j(t) H_j$ the Hamiltonian of a bilinear single-control problem parametrized by a vector $\boldsymbol{\phi}$. The LCU-GRAPE  circuit shown in (b) is an estimator of the (renormalized) gradient given in (a) and therefore implements, on a quantum circuit, the gradient sampling procedure outlined in Refs.~\cite{ kottmann2023evaluating, Leng2024}. The operation $W^{\text{LCU}}_{\boldsymbol{u}^{\boldsymbol{\phi}}}$ loads the LCU register that uses $r=\lceil \log(N_T) \rceil$ qubits (we consider again for the sake of simplicity only linear combinations of positive coefficients) with  coefficients proportional to $\pdv{u_j(t_i)}{\phi_{ij}}\Delta t_i$ and can therefore The derivatives with respect to the single time-sliced control pulses can be implemented in different way, either with the SU($d$) gradient circuit -- see Fig.~\eqref{fig:sun_gradient_circuit} and the procedure outlined in Section \ref{sec:sun_gradient_lcu}, or with another type of gate gradient estimation method, depending on the structure of the spectrum of the gate element $V_{\boldsymbol{\phi}}(t_i), i=1,...,N_T$ itself. We also see that this gradient resembles the structure of a forward derivative, because each time-sliced gate depends on the same set of parameters $\boldsymbol{\phi}$, so the gradient results a sum of each contribution. Overall, this procedure gives a $O(T^2/\epsilon^2)$ and $O(T/\epsilon)$ sampling complexity for non-amplified and amplified estimates, respectively.}
    \label{fig:grape_circuit}
\end{figure*}

The formal derivation of the exact average of the square truncation error with respect to Hermitian matrices $H_0$ and $H_1$ sampled from the GUE is given in Appendix \ref{sec:truncation_error}. The asymptotic behaviour of the squared truncation error is given by:
\begin{equation}
  \EX[R_L^2] \sim O\left(\frac{\Lambda_1^2}{d^2} \left(\|\btheta\|\Lambda_0 \right)^{2(L+1)} \right),
\end{equation} 
where $d=2^n$, $\Lambda_1$ and $\Lambda_0$ are the Dyson coefficients of the distributions of $H_0$ and $H_1$, and $\btheta$ is the parameter vector of the control problem as defined in Eq.~\eqref{eq:simple_control_problem}. The value $\EX[R_L^2]$ can also provide us with a bound for the average error itself since $\abs{\EX[R_L]} \leq \sqrt{\EX[R_L^2]
} $. We simulate the result by computing the gradient approximation in Eq.~\eqref{eq:sun_lthorder} with the problem structure given in Eq.~\eqref{eq:h0h1_grad}, i.e., with $\theta_1 = 0$ and $\theta_0 =  1$. We average the approximation error over 50 random Hamiltonians $H_0$ and $H_1$ sampled from the GUE with $\Lambda_0 = \Lambda_1 = 1$ (see Ref.~\cite{Tao2012-randommatrix} and Appendix \ref{sec:random_operators}) for different numbers of qubits and values of the parameter $\theta$.
The results of our simulation of the SU($d$) gradient are shown in Fig.~\ref{fig:sun_gradient_sim} for different test values of $\theta$. The gradient used for comparison is obtained with \textsc{JAX} \cite{jax2018github}. Full lines and dashed lines represent the simulated average of $R_L^2$ -- POTQ gradient -- (a) and $\mathcal{R}_L^2$ -- infidelity gradient - (b) up to an order $L$. Dotted lines represent the symmetric finite-difference (FD) approach with a shift $\delta = 0.75$ -- see also Ref.~\cite{Wiersema2024herecomessun}. We see that the theoretical predictions match the average value of the squared truncation error within the standard deviation represented by the shaded regions and that the approximate gradient to order $L$ quickly reaches the precision of a symmetric $O(1)$ FD approach, which is commonly used when sampling classically from quantum circuits \cite{Wiersema2024herecomessun}. 

By setting an appropriate maximum expansion index $L$, we can study the convergence of the gradient approximation. The convergence with respect to $L$ is particularly relevant for the LCU implementation, because, if the integer $L$ needed to reach a gradient precision $\delta$ is reasonably small, so will be the depth of the corresponding LCU-based quantum circuit. 

In Fig.~\ref{fig:sun_gradient_sim} (c) we show the maximum norm of the unitary gradient approximation:
\begin{align}\label{eq:infimum_norm}
d^{\infty}_{L} = \norm{\pdv{V(\boldsymbol{\theta})}{\theta_1} \Big |_{\boldsymbol{\theta} = (1, 0)^T} - \mathcal{W}^c_L}_{\infty},
\end{align}
for the maximum index considered, i.e., for us $L=14$. as a function of the parmeter $\Delta t$. This plot helps us visualize the precision of the unitary gradient itself if compared to the precision of the gradient of the two quantum cost functions given in Fig.~\ref{fig:sun_gradient_sim} (a) and (b). In Fig.~\ref{fig:sun_gradient_sim} (d) we see as an example the change of $R_L^2$ for $L=1$ as a function of $\theta$: the behaviour of the function reflects the predicted dependence of $\EX[R_L^2]$ from $\theta$, i.e., $O[\theta^{2(L+1)}]$, as shown in Eq.~\eqref{eq:approx_RL2} and in more detail in Eq.~\eqref{eq:exact_RL2}.

\subsection{Gradients of quantum dynamics}\label{sec:q_control_gradients}
In the case of pulse-level optimization, the problem usually has a bilinear structure with a drift Hamiltonian $H_0$ and control Hamiltonians $H_1, ..., H_{K-1}$ -- see also Eq.~\eqref{eq:sun_gate} -- with corresponding time-dependent control fields $u^{\boldsymbol{\phi}_1}_1(t), ..., u^{\boldsymbol{\phi}_{K-1}}_{K-1}(t)$. W.l.o.g., we can assume that these control fields are each parametrized independently by real vectors $\boldsymbol{\phi}_1, ..., \boldsymbol{\phi}_{K-1}$. The gradient of the unitary evolution of $H_{\boldsymbol{\phi}}(t) = H_0 + \sum_{k=1}^{K-1} u^{\boldsymbol{\phi}_k}_k(t) H_k$, where $\boldsymbol{\phi} = (\boldsymbol{\phi}_1, ..., \boldsymbol{\phi}_{K-1})^{\text{T}}$, that is of $V_{\boldsymbol{\phi}}(t) = \mathcal{T}e^{-i \int_{t_0}^T H(\tau) d\tau}$, where $\mathcal{T}$ is the time-ordering operator, from the initial time $t=t_0$ to the final time $t=T$ with respect to the parameters, is given by \cite{Khaneja2005}:
\begin{align}\label{eq:dU_integral}
    \pdv{V_{\boldsymbol{\phi}}}{\phi_{ij}} = - i \int_{t=t_0}^T V_{\boldsymbol{\phi}}(t_0, \tau) H_j \pdv{u^{\boldsymbol{\phi}_j}_j(t)}{\phi_{ij}} V_{\boldsymbol{\phi}}(\tau, T) d \tau.
\end{align}
For a control cost function of the type of Eq.~\eqref{eq:SEst}, i.e.:
\begin{align}\label{eq:control_cost_function}
    C(\boldsymbol{\phi}) = \tr{V_{\boldsymbol{\phi}}(t_0, T) \rho V^{\dagger}_{\boldsymbol{\phi}}(t_0, T) \mathcal{O}},
\end{align}
the gradient is given by -- see Appendix \ref{appx:QC_grad} and Refs.~\cite{Goerz2015, kottmann2023evaluating, Leng2024}:
\begin{align}\label{eq:grape_continuous_grad}
    &\pdv{C(\boldsymbol{\phi})}{\phi_{ij}} = -i \int_{t=t_0}^T s^{\boldsymbol{\phi}}_j(\tau) \pdv{u^{\boldsymbol{\phi}_j}_j(t)}{\phi_{ij}} d \tau, \\
    &s_j^{\boldsymbol{\phi}}(\tau) = \tr{\rho(T) [H_j(t_0, \tau), \mathcal{O}]}.
\end{align}
Eq.~\eqref{eq:grape_continuous_grad} resembles the implementation of the analog LCU algorithm \cite{Chakraborty2024implementingany}. Full quantum control gradients are usually too challenging to implement for basic qubit optimization on near-term devices. However, quantum control has been suggested as a possible ansatz to optimize variational quantum algorithms \cite{Magann2021}. In these implementations, the use of LCU algorithms can be considered beneficial. Usually, the gradient given in Eq.~\eqref{eq:grape_continuous_grad} is computed for a discretized time dynamics, i.e., where $T = \sum_{n=1}^{N_T} \Delta t_n$, in which case we have:
\begin{align}\label{eq:grape_sum}
    \pdv{C(\boldsymbol{\phi})}{\phi_{ij}} \approx -i \sum_{n=1}^{N_T}  s^{\boldsymbol{\phi}}_j(t_n) \pdv{u^{\boldsymbol{\phi}_j}_j(t_n)}{\phi_{ij}} \Delta t_n,
\end{align}
which is the discretized version of the GRAPE gradient \cite{Khaneja2005}. We refer to the circuit estimator for this quantity as LCU-GRAPE. Assuming the variance of the control operators is bounded, the sampling complexity of this approach with a classical sampling process scales as $O(T^2/\epsilon^2)$ and therefore as $O(T/\epsilon)$ for an amplified sampling. The circuit given in Fig.~\ref{fig:grape_circuit} (LCU-GRAPE), when combined with amplitude estimation, potentially allows for quadratic speed up in the evaluation of GRAPE-like gradients. However, the depth of such LCU circuits and the additional ancillas makes them impractical on near-term quantum devices and in the control of experimental qubits for gate synthesis and state preparation \cite{Kelly2014}. On the other hand, Ref.~\cite{Leng2024} draws connections between (Ordinary Differential Equation Networks) ODENets \cite{chen2018neural} and quantum dynamics typical of control systems. The analysis contained therein provides us with training methods for ODENets on quantum systems using Eq.~\eqref{eq:grape_continuous_grad} and stochastic parameter-shift rules \cite{Banchi2021measuringanalytic}. As we discussed above, classical estimation of such a derivative is quadratic in time -- $O(T^2/\epsilon^2)$ -- both in the case of SE and LCU, while amplified LCU-GRAPE can potentially reach $O(T/\epsilon)$. For large parametrized quantum models, this represents a significant speedup. While its relevance is limited for NISQ circuits, it will most probably increase in the next years as logical error rates continue to decrease. 

\section{Summary and Conclusion}
In this work, we thoroughly analyze the sampling complexity of different LCU estimators in different contexts. In the first part, we focused on reviewing the approaches to LCU sampling and estimation compared to SE. The complexity of estimating observables using LCU and SE  without any quantum AE routine is the same, i.e., $O(L^2/\epsilon^2)$ \cite{Chakraborty2024implementingany}. In the case of AE, the sampling complexity of the LCU estimator reduces to $O(L/\epsilon)$, which is considerably faster than the $O(L\sqrt{L}/\epsilon)$ scaling of SE. We also draw connections between classical probability theory, circuit sampling and DQC1 by considering the SA-LCU approach presented in Ref.~\cite{Chakraborty2024implementingany}. 
In the second part of the work, we focus on the specific application to gradient estimation, and more specifically we discuss how LCU can be used to represent the gradient of arbitrary cost functions that depend on unitary evolution operators, such as the SU($d$) gradient introduced in Ref.~\cite{Wiersema2024herecomessun} and GRAPE-like gradients (LCU-GRAPE). We present and discuss in particular the circuits that allow for the estimation of SU($d$) gradients and control gradients, and analyze the convergence properties of the gradient approximation both from a numerical and an analytical perspective using concepts from random matrix theory. These results are relevant for advanced implementations of gradient-based optimization on quantum hardware that make use of either non-standard multi-qubit gates or circuit ansätze based on quantum optimal control theory.\\

\section{Code and Data Availability}
The code and data used for this work are available at Ref.~\cite{GPVQuEst}. 
\begin{acknowledgments}
This work was supported by AIDAS, The European Joint Virtual Lab, an initiative of Forschungszentrum Jülich (FZJ) and the French Alternative Energies and Atomic Energy Commission (CEA), by the German Federal Ministry of Education and Research (BMBF), project QSolid, Grant No.~13N16149, by the German Research Foundation (DFG) under Germany’s Excellence Strategy – Cluster of Excellence Matter and Light for Quantum Computing (ML4Q) EXC 2004/1 – 390534769 and by the Jülich Supercomputing Center (JSC). We acknowledge funding from the Horizon Europe program (HORIZON-CL4-2021-DIGITAL-EMERGING-02-10) via the project 101080085 (QCFD) and by HORIZON-CL4-2022-QUANTUM-01-SGA Project
under Grant 101113946 OpenSuperQPlus10. We are grateful to Markus Heinrich, Manuel Guatto, Robert Zeier and Roberto Gargiulo for the stimulating discussions. Simulations were realized in \textsc{python} using the libraries \textsc{qiskit} \cite{qiskit2024}, \textsc{jax} \cite{jax2018github}, \textsc{qclib} \cite{QCLib} and \textsc{numba} \cite{lam2015numba}. 
\end{acknowledgments}

\onecolumngrid
\appendix
\section{Trace estimation with DQC1}

\subsection{DQC1}

When dealing with NMR quantum computers \cite{Jones2010}, it is common to have access to $n$-qubit mixed states. On these and similar quantum systems, having a quantum computer that only uses mixed states and some control qubits would be beneficial, rather than using only pure states that undergo unitary evolutions. It turns out that there exist problems that can be solved exponentially faster on this type of quantum computer compared to a classical computer \cite{Shor2008}. However, it has been shown that this computation model is also significantly weaker than standard quantum computation \cite{Shor2008, Shepherd2006computation}. This model is referred to nowadays as DQC1 \cite{Knill1998} (Deterministic Quantum Computation with One Clean Qubit). Several problems involving computing distance measures between unitaries and states using quantum circuits are DQC1-complete or -hard \cite{Poulin2004, Khairy2020, Bravo-Prieto2019}. One for all, trace estimation of unitaries is DQC1-complete \cite{Shepherd2006computation}. The same is true for energy estimation for Hamiltonians with very low (logarithmic) connectivity \cite{Shepherd2006computation}.
It seems that problems involving LCU-based sampling tend to be DQC1-hard \cite{Bravo-Prieto2019}. In general, one can construct a DQC1-cost function by implementing a LCU-type circuit with controlled operations \cite{Park2019}. 

\subsection{Sampling Unitary traces and DQC1 basics}\label{sec:samplingtraces}

The complexity class DQC1 that allows direct computation with one clean qubit uses a single- or multi-qubit controlled unitary $V$ acting on a maximally mixed state (which corresponds to uniformly sampled random pure states) to perform quantum computation \cite{Knill1998}. This model has been shown to be hard to simulate on classical computers if more than three output qubits are considered \cite{Morimae2014}. Estimating the re-normalized real part of the trace of a unitary $V$ is a complete problem for DQC1 \cite{Shepherd2006computation}. The circuit that allows this estimation is given in Fig.~\ref{fig:dqc1-rev} (a) for mixed-states inputs and (b) for pure states. \\
The probability of measuring the control qubit for the circuit that uses mixed states -- see Fig.~\ref{fig:dqc1-rev} (a) -- in the zero (+) or one (-) state is given by
\begin{align}
    \bar{p}_{\pm} = \frac{1}{2}\left(1 \pm \frac{1}{d} \Re{\Tr{V}}\right).
\end{align}
The imaginary part can be obtained by inserting an $S$ gate before the final measurement in the circuits given in Fig.~\eqref{fig:dqc1-rev}.
We can see that if we want to compute the trace of a unitary, the variance of the estimator for the trace behaves as a Bernoulli variable \cite{Aharanov2006}:
\begin{align}
    &\text{Var}(\bar{x}) = \bar{p}_{+}(1 -\bar{p}_{+}).
\end{align}
Here, we consider only one of the two outcomes and thus set $ \bar{p} = \bar{p}_{+}$. Without using a maximally mixed state $\rho = \frac{\mathbb{I}}{d}$, computing the trace requires preparing $d=2^n$ orthonormal basis states $\{ \ket{i} \}_{i=1}^d$ and estimating their corresponding probability distribution $p_i$:
\begin{align}
    p_i = \frac{1}{2}(1 + \Re{\bra{i} V \ket{i}}).
\end{align}
As the real part of the trace of $V$ is given by $\Re{\Tr{V}} = \sum_{i=1}^{d} \Re{\bra{i} V \ket{i}}$, we can construct an estimator for the trace by creating an estimator $\bar{x}'$ that collects the counts of the $d$ different circuits and sums them together. \\
The variance of this estimator behaves as a sum of independent Bernoulli variables, each one with variance $p_i(1 - p_i)$:
\begin{align}
    \text{Var}(\bar{x}') = \frac{1}{d^2} \sum_{i=1}^d p_i(1 - p_i),
\end{align}
and is $d$-times smaller than the variance of the estimator based on the maximally mixed state.
\begin{proof}
We prepare the $d$ states $\ket{i}, i=1,...,d$ and apply the Hadamard test on each of them using the controlled unitary $V$. Since the states are prepared on $d$ different Hilbert spaces, the variance of the estimates is simply the sum of the variances, each one equal to $p_i(1 - p_i)$ for $i=1,...,d$.
Afterwards, we use the fact that $p_i \leq \sqrt{p_i}$ for $0 \leq p_i \leq 1$ and write
\begin{align}
    \frac{1}{d^2} \sum_{i=1}^d \sum_{j=1}^d \left( \EX \left[x_i x_j \right] - p_i p_j \right) \overset{\text{C.S.}}{\leq} \frac{1}{d^2}  \sum_{i=1}^d \sum_{j=1}^d  \sqrt{p_i} \sqrt{p_j} (1 - \sqrt{p_i} \sqrt{p_j}).
\end{align}
Then we apply one of the generalized mean inequalities (GM) \cite{Bullen2003}, which follows from the Jensen's inequality \cite{Jensen1906}:
\begin{align}
      \sum_{i=1}^d \sum_{j=1}^d  \sqrt{p_i} \sqrt{p_j} \overset{\text{GM}}{\leq}  \frac{1}{d} \sum_{i=1}^d  p_i,
\end{align}
and we obtain
\begin{align}
    \frac{1}{d^2} \sum_{i=1}^d \sum_{j=1}^d \left( \EX \left[x_i x_j \right] - p_i p_j \right) \leq  \frac{1}{d} \sum_{i=1}^d  p_i \left(1 -  \frac{1}{d} \sum_{i=1}^d  p_i \right) = \text{Var}(\bar{x}).
\end{align}
The last term is the variance of the DQC1 estimator. This derivation assumes that the circuits are somewhat correlated. However, if no correlations are present, because the estimators are prepared independently from each other, the term on the left hand side is $d$ times smaller than the full variance with non-zero covariances. Therefore, we have
\begin{align}
     \text{Var}(\bar{x}') \leq  \frac{1}{d} \text{Var}(\bar{x}).
\end{align}
\end{proof}
We see that the variance of the first estimator is at least $d$ times smaller than the one of the DQC1 estimator. However, the DQC1 estimator does not require the preparation of $d$ different states, as long as we have access to a source of randomly distributed pure states in input (or, equivalently, a maximally mixed state). In summary, the query complexity of both  estimators with precision $\epsilon$ is $O(d^2/\epsilon^2)$ -- in the former case, we need to sample the sum of $d$ circuit outputs with precision $O(d/\epsilon^2)$, in the latter we sample from one circuit with precision $O(d^2/\epsilon^2)$.

\begin{figure*}
\includegraphics[width=\textwidth]{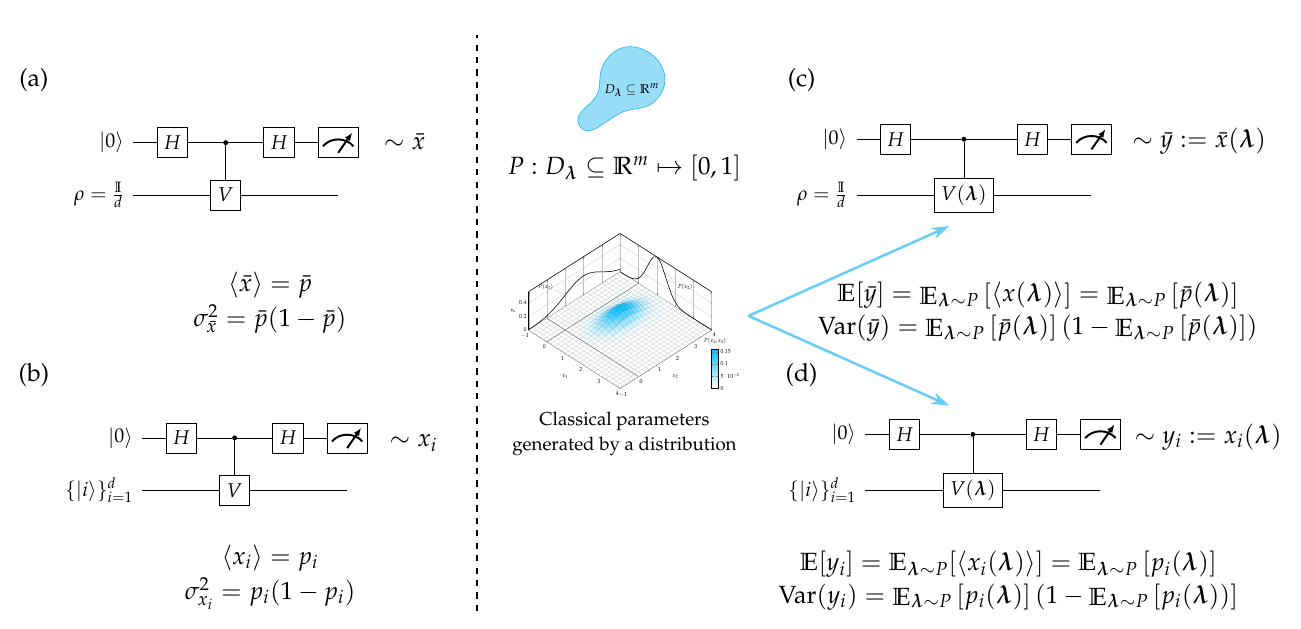}
    \caption{In this figure we consider different types of sampling problems that use the Hadamard-like test circuit normally employed for DQC1 tasks. In (a) and (b) the same circuits are considered with a fixed unitary $V$, without any external dependance. The values $\bar{x}$ and $x_i, i=1,...,d$ represent the Bernoulli counts sampled using the circuits. The circuits (c) and (d) instead depend for some internal degrees of freedom on an external parameter $\boldsymbol{\lambda} \in D_{\boldsymbol{\lambda}} \subseteq \mathbb{R}^m$ sampled from a probability distribution $P: D_{\boldsymbol{\lambda}} \subseteq \mathbb{R}^m \mapsto [0,1]$. In this case, $\bar{y}$ and $y_i, i=1,...,d$ represent the new Bernoulli counts sampled using both the circuits and the distribution $P$. Therefore, sampling from these circuits requires averaging over the number of shots sampled from the circuits and the number of samples obtained from the probability distribution $P$. In (c), the circuit uses a maximally mixed state, whereas in (d) the $d$ basis states $\ket{i}, i=1,...,d$ are prepared and the corresponding probability distribution of measuring the control qubit in state $0$ or $1$ is sampled. }
    \label{fig:dqc1-rev}
\end{figure*}

\subsection{Conditional sampling and averaging}\label{sec:sampling_PoissonBinomial}
The sampling methods discussed in the previous sections offer some insight in the computational structure of these methods. In particular, we see that an LCU circuit (or the equivalent ancilla-free approach) maps an originally Bernoulli-like estimation problem to an equivalent Bernoulli-like estimation problem. If the original goal is to sum $L$ Bernoulli estimates $p_1, ..., p_L$, whose value lies between $0$ and $L$, the LCU circuit outputs the renormalized version of this estimate, whose value lies between $0$ and $1$. As a consequence, the corresponding non-renormalized estimate $L \bar{p}$ behaves as a sum of strongly correlated variables with variance $L^2 \bar{p}(1 - \bar{p})$. Eq.~\eqref{eq:max_variance} shows that the correlations induced by the LCU circuits are stronger than any classical correlation. In classical probability theory, the sum of $L$ independent Bernoulli-distributed trials follows the Poisson-Binomial distribution \cite{Wang1993}. Interestingly, this distribution can be approximated by a Binomial distribution with mean $\bar{p}$. However, the variance of such distribution is always $L$ times larger than the true variance of the Poisson-Binomial distribution. The output of the LCU circuit corresponds exactly to the approximation of the Poisson-Binomial distribution, which instead is the distribution of $L$ independent Hadamard tests. A different question arises when we sum potentially correlated variables, for example if the originally independent Hadamard tests are sampled from a joint distribution. In this case, there may be potential correlations present induced, e.g., by random circuit sampling \cite{heinrich2023randomizedbenchmarkingrandomquantum}, due to the law of total variance.

\section{SE with near-term amplitude estimation}\label{appx:SE_AE}

Before moving to to SE with amplitude amplification, we briefly review the principles of Maximum Likelihood Quantum Amplitude Estimation (MLQAE) \cite{Suzuki2020}. MLQAE starts from the amplified states:
\begin{align}\label{eq:psi_m}
    \ket{\psi^{m}} = \mathcal{Q}^{m} \ket{\psi} = \sin((2m+1)\theta_p) \ket{1} \ket{\psi_1} + \cos((2m+1)\theta_p) \ket{0} \ket{\psi_2},
\end{align}
where $\sin(\theta_p)^2 = p$, where $p$ is the probability of measuring $\ket{1}$ and also the parameter to estimate, using the language of estimation theory. The states in Eq.~\eqref{eq:psi_m} are obtained after applying $m$ times the Grover operator of AE \cite{Brassard2002}: $\mathcal{Q} = - \mathcal{V} S_0 \mathcal{V}^{-1}S_{\chi}$, as defined in Eq.~\eqref{eq:Q_operator}. If we prepare $m$ copies of the initial state $\ket{\psi}$ and perform independent runs of AE each one with one of the operators $\mathcal{Q}, \mathcal{Q}^2, ..., \mathcal{Q}^M$ respectively, the probability distribution of this estimation problem is given by
\begin{align}
    P(\theta_p, \boldsymbol{x}, \boldsymbol{N}) = \prod_{i=1}^M \prod_{k=1}^{N_m} \sin^{2}((2m+1)\theta_p)^{x_{mk}} \cos^{2}((2m+1)\theta_p)^{1 - x_{mk}},
\end{align}
where the integers $N_m$ are the number of shots from each circuit implementing $\mathcal{Q}, \mathcal{Q}^2, ..., \mathcal{Q}^M$ and $x_{mk}$, $k=1,...,N_m,\ m=1,...,M$ represents the data used for parameter estimation, i.e., the binary measurement results of the shot $k$ from the $m$th circuit where the state is found in $\ket{1}$ (while $1 - x_{mk}$ is the outcome corresponding to $\ket{0})$. All data is represented by the vector $\boldsymbol{x}=(x_{11},...,x_{MN_M})^{\text{T}}$ and the discrete parameters of the distribution are given by the vector of integers $\boldsymbol{N}=(N_1,...,N_M)^{\text{T}}$.
Ref.~\cite{Suzuki2020} shows that if one applies the principle of Maximum Likelihood Estimation on $P$ and if the allocation of the number of shots is chosen optimally, the total number of queries $n_q = \sum_{m=1}^M N_m$, which accounts for both the shots and the repeated applications of the operator $Q$, needed can potentially reach the Heisenberg limit, that is, the Fisher information of the Maximum Likelihood estimator obeys the following bound \cite{Suzuki2020}:
\begin{align}
    F \leq \frac{n_q^2}{p(1-p)}.
\end{align}

Now we consider again the problem of estimating outputprobabilities $p_1, ..., p_L$ from quantum circuits $\mathcal{V}_i, ..., \mathcal{V}_L$ defined as in Eq.~\eqref{eq:U_aest}. Each of these circuits resembles a Hadamard test. We define the random variables corresponding to the \textit{good} outputs of each circuit as $X_1, ..., X_L$, where $\EX[X_i] = p_i$, for $i=1,...,L $. Let us approach the problem from the point of view of MLQAE. If we assume that each estimate is obtained through MLQAE and that the estimator reaches the Heisenberg limit, then each estimator $\tilde{p}_1, ..., \tilde{p}_L$ that uses either near-term AE \cite{Suzuki2020} has an error proportional to: 
\begin{align}
    \epsilon_i^2 = \frac{1}{F(X_i)} = \frac{\sigma_i^2}{[n_q^{(i)}]^2},
\end{align}
where $\sigma_i^2 = p_i(1-p_i)$ -- $F$ is the Fisher information of the $i$th random variable $X_i$ associated with $p_i$ -- instead of the classical $\epsilon_i^2 = \frac{\sigma_i^2}{n_q^{(i)}}$. However, if these amplitudes are summed classically, the uncertainty propagation obeys again the rules of classical propagation of the uncertainty. This is due to the Stam inequality, which argues that for $L$ random variables $X_1, ..., X_L$ the Fisher information of any given parameter obeys the following bound \cite{Stam1959, Bobkov2014}:
\begin{align}
    \frac{1}{F(X_1 + ... + X_L)} \geq \frac{1}{F(X_1)} + ... + \frac{1}{F(X_L)}.
\end{align}
As a consequence, for the sum of amplified estimates with Heisenberg-like scaling of the classical Fisher information -- see Ref.~\cite{Suzuki2020} --, we have the following bound for the error:
\begin{align}
    \epsilon^2 \geq \frac{1}{F(X_1 + ... + X_L)} \geq \sum_{i=1}^L w_i^2 \frac{\sigma_i^2}{[n_q^{(i)}]^2}.
\end{align}
The total number of queries from all circuits can be minimized with respect to $n_q^{(i)}, \ i=1,...,L$ using Lagrange multipliers \cite{Rubin2018}:
\begin{align}
    &\mathcal{L}(n_q^{(i)}, ..., n_q^{(L)}, \lambda) = \sum_{i=1}^L n_q^{(i)} + \lambda \left(\sum_{i=1}^L \frac{w_i^2 \sigma_i^2 }{[n_q^{(i)}]^2} - \epsilon^2\right) \\
    &\hspace{1.0cm} \left( \forall 1 \leq i \leq L:\ \pdv{\mathcal{L}}{n_q^{(i)}} = 0 \right) \wedge \left( \pdv{\mathcal{L}}{\lambda} = 0 \right). 
\end{align}
The minimization procedure results in:
\begin{align}
    &\hspace{-0.3cm}\forall 1 \leq i \leq L:\ [n_q^{(i)}]^3 = 2\sigma_i w_i^2 \lambda \\
    &2\lambda = \left(\frac{1}{\epsilon^2}\sum_{i=1}^L (w_i^2 \sigma_i)^{\frac{1}{3}} \right)^{\frac{3}{2}},
\end{align}
and leads to a total number of queries $N_q$ equal to
\begin{align}\label{eq:ae_te_complexity}
    N_q = \sum_{i=1}^L n_q^{(i)} = \frac{1}{\epsilon} \sum_{i=1}^L  \left( \sum_{i=1}^L (2\sigma_i w_i^2)^{\frac{1}{3}} \right)^{\frac{1}{2}} \sim O \left(\frac{L \sqrt{L}}{\epsilon}\right).
\end{align}
As expected, we obtain a partial improvement over the classical evaluation even by implementing MLQAE. However, the improvement is worse than a LCU circuit where MLQAE is applied. It seems that the typical LCU entanglement is needed in order to obtain a better improvment. It is unclear whether another MLQAE approach could obtain an asymptotical sampling complexity of $O(L/\epsilon)$ without the use of an LCU circuit. If the estimation is divided on batches of size $k$, which are encoded in corresponding LCU circuits and subsequently amplified, and whose results are then summed classically, one has a complexity of $O(\frac{L}{\epsilon}\sqrt{L/k})$, which reduces to both the LCU and SE amplfied case for $k=1$ and $k=L$, respectively.

\subsubsection{Classical shadow tomography for quantum cost functions}
To better contextualize the use of the LCU in estimation, we want to compare it to other popular estimation strategies that emerged in the last years: Tomography based on classical shadows \cite{Huang2020, Huang2020shadow}. A classical shadow $\hat{\rho}$ is an estimator of the true density matrix of a quantum experiment $\rho$. It can be used to estimate the mean value of any observable with respect to $\rho$. In facts, for any channel acting on $\rho$ -- $d=2^n$ for $n$ qubits --, we can write \cite{Huang2020shadow, Mele2024introductiontohaar}:
\begin{align}
    \mathcal{Q}(\rho) = \sum_{l=1}^d \mathbb{E}_{U \sim \nu} \left[ \tr{ U  \rho U^{\dagger} \ket{l} \bra{l}} U^{\dagger} \ket{l} \bra{l} U \right],
\end{align}
where $U \sim \nu$ denotes random sampling from the Haar measure. Then the estimator $\tilde{\rho}$
\begin{align}
    \tilde{\rho} = \mathcal{Q}^{-1}(U^{\dagger}  \ket{l} \bra{l} U)
\end{align}
is an unbiased estimator of $\rho$ \cite{Mele2024introductiontohaar}, which can be used to estimate mean values of arbitrary observables. For an observable with $L$ non-zero non-commuting Pauli coefficients, the sampling complexity is $O[L \log(L)/\epsilon^4]$. The significant reduction in the number of queries needed is traded off with a worse scaling with the precision $\epsilon$. In the presence of AE sampling, as we see from Eq.~\eqref{eq:ae_te_complexity} -- see also \cite{Mele2024introductiontohaar} --, the variance of the estimator reduces to $O(\sqrt{L}/\epsilon)$, which results in $O[\sqrt{L}\log(L)/\epsilon^3]$.
In order to be performed efficiently, classical shadows have two requirements: (a) the inversion operation of the channel needs to be efficient \cite{abbas2023quantumbackpropagationinformationreuse} and (b) an efficient method to sample Haar random unitaries $U \sim \nu$ needs to be available. Both the classical shadow approach and the approach that uses the Jordan algorithm speed up the oracular evaluation of multiple observables, but do not reduce the variance. 
If the estimation of a quantum cost function is carried out using the shadow estimation protocol, the estimation of its gradient can, in theory, also be carried out using the same method \cite{abbas2023quantumbackpropagationinformationreuse}. General gradient estimation of quantum cost functions based on LCU, finite-difference or PSR approaches has a linear dependence from the number of variational parameters $N$, i.e., $O(N)$. This represents a significant bottleneck to the efficient implementation of variational circuits optimization \cite{Cerezo2021review} and quantum machine learning algorithms \cite{Biamonte2017, Schuld2020}. In this case, we know that estimating the gradient of a quantum cost function with observable $\mathcal{O}$ that evolves with a unitary $V_i = \exp{-i H_i \theta_i}$ corresponds to estimating the mean value of the observable $i\left[\mathcal{O}, H_i \right]$. For a sequence of independent gates with the same structure, a nested estimation is required in which $i\left[\mathcal{O}_i , H_i \right]$, $\mathcal{O}_i = \prod_{j=1}^i V_i \rho \left( \prod_{j=1}^i V_i \right)^{\dagger}$ fully determines the gradient component $\pdv{}{\theta_i} C(\boldsymbol{\theta_i})$ for $i=1,...,N$. Ref.~\cite{abbas2023quantumbackpropagationinformationreuse} shows that no general backpropagation can be accomplished without access to a quantum memory. However, there seem to be classes of circuits that are not classically simulable and whose gradients can be sampled in sub-linear time with respect to the number of parameters $N$ \cite{ bowles2024backpropagationscalingparameterisedquantum}.

\subsection{SU($d$) gradient with coherent state}\label{sec:alpha_gradient}
We consider the cost function $C(\boldsymbol{\theta}) = \tr{V_c(\boldsymbol{\theta}) \rho V_c(\boldsymbol{\theta})^{\dagger} Z_{\text{prod}}}$. A generalization to arbitrary observables can be achieved by implementing either the SE or the LCU estimator. A generalization to positive and negative linear combinations is given in Section \ref{sec:extension_glc}. The gradient of the cost function is provided in Eq.~\eqref{eq:cost_function_grad} for a unitary $V_c(\boldsymbol{\theta}) = e^{-i (H_0 \theta_0 + H_1 \theta_1) \Delta t}$ -- see Eq.~\eqref{eq:simple_control_problem}. The derivative with respect to $\theta_1$ evaluated at $\theta_1=0$ is given by:
\begin{align}\label{eq:lcu_sungrad_derivation}
    \pdv{C(\boldsymbol{\theta})}{\theta_1} \Big |_{\boldsymbol{\theta} = (1, 0)^T} =  \tr{\pdv{V_c(\boldsymbol{\theta})}{\theta_i} \rho V_c^{\dagger}(\boldsymbol{\theta}) Z_{\text{prod}}} = 2\norm{\boldsymbol{a}}_1 \sum_{l=0}^{\infty} \Re{\frac{(- 2 i \Delta t)^{l}}{(l+1)!}\tr{\left(\frac{1}{2}\right)^l\ad^{l}_{H_0}(H_1) \rho_c(\boldsymbol{\theta}) Z_{\text{prod}}}}
\end{align}
where on the right hand side we used the substitution $\rho_c(\boldsymbol{\theta}) = V_c(\boldsymbol{\theta})\rho V_c^{\dagger}(\boldsymbol{\theta})$. We note that for a complex value $z \in \mathbb{C}$:
\begin{align}
    \Re{i^l z} = \begin{cases} 
      -\Im{z} & l\mod 4 = 1\\
      -\Re{z} & l\mod 4 = 2 \\
      \Im{z} & l\mod 4 = 3\\
      \Re{z} & l\mod 4 = 0,
   \end{cases}
\end{align}
so we cannot simply move the $i^l$ outside of the real part operation. 
The adjoint term is computed by the circuit given in Fig.~\ref{fig:sun_gradient_circuit} (b) for a given order $l$. Furthermore, we use a Hadamard test-like circuit. For a renormalized observable $Z_{\text{prod}}$ whose mean value lies in $I=(-1,1)$, we can always use LCU and a Hadamard test \cite{li2024efficientquantumgradienthigherorder, Oshio2024} to estimate $\frac{1}{2}\left(1 \pm \Re{\rho_c(\boldsymbol{\theta}) Z_{\text{prod}}}\right)$ (using the Hadamard gate $H$) and/or $\frac{1}{2}\left(1 \pm \Im{\rho_c(\boldsymbol{\theta}) Z_{\text{prod}}}\right)$ (using the Hadamard gate and the $S$ gate, $HS$). 
We introduce a multi-controlled register for the LCU summation that goes from $0$ to $L -1$ with $r = \lceil \log(L) \rceil$ qubits :
\begin{align}
    \ket{\sqrt{2 \Delta t}_L} = \frac{1}{\sqrt{\mathcal{M}_{L}}} \sum_{l=0}^{L-1} \frac{(\sqrt{2 \Delta t})^{l}}{\sqrt{(l)!}} \ket{l},
\end{align}
where $\mathcal{M}_{L} = \sum_{l=0}^{L-1} \frac{\abs{2 \Delta t}^l}{l!}$ is the renormalization factor. We then use a Hadamard test circuit with $\bar{Z_{\text{prod}}}$ as measurement and controlled commutator circuits $l=0,...,L-1$ as in Fig.~\ref{fig:sun_gradient_circuit} (b), we can estimate the quantity:
\begin{align}\label{eq:lcu_sungrad_estimation_L}
     q^L_{\pm} = \frac{1}{2}\left(1 \pm \frac{1}{ \mathcal{M}_L}\sum_{l=0}^{L-1}  \frac{(2 \Delta t)^{l}}{l!} \Re{\tr{(-i/2)^l \ad^{l}_{H_0}(H_1) \rho_c(\boldsymbol{\theta})Z_{\text{prod}}}}\right),
\end{align}
In the limit of $L$ approaching infinity, we have the coherent state:
\begin{align}
    \ket{\sqrt{2\Delta t}} = e^{-\abs{\Delta t}} \sum_{l=0}^{\infty} \frac{(\sqrt{2\Delta t})^{l}}{\sqrt{(l)!}} \ket{l},
\end{align}
and
\begin{align}
    \label{eq:lcu_sungrad_estimation_inf}
     q^{\infty}_{\pm} = \frac{1}{2}\left(1 \pm \frac{1}{ \mathcal{M}_{\Delta t}}\sum_{l=0}^{\infty}  \frac{(2 \Delta t)^{l}}{l!} \Re{\tr{(-i/2)^l \ad^{l}_{H_0}(H_1) \rho_c(\boldsymbol{\theta})Z_{\text{prod}}}}\right),
\end{align}
where $\mathcal{M}_{\Delta t} = \exp{ 2 \abs{\Delta t}}$ is the normalization coefficient of the coherent state. The estimate given by the commutator circuit is already renormalized thanks to the $l$ Hadamard gates in Fig.~\ref{fig:sun_gradient_circuit}. The addition of a number $l$ of $S$ gates provides the $(-i)^l$ terms.
Eq.~\eqref{eq:lcu_sungrad_estimation_inf} does not match Eq.~\eqref{eq:lcu_sungrad_derivation} yet. In order to match the two equations up to a constant, we change the first controlled operation on the LCU register to be the identity, whereas the operator that corresponds to $\ad^{l}_{H_0}(H_1)$ is conditioned on the next entry in the register. As a result we have:
\begin{align}
    e^{\infty}_{\pm} = \frac{1}{2} \left[ 1 \pm \frac{1}{\mathcal{M}_{\Delta t}} \left(b(\btheta) + \sum_{l=1}^{\infty} \frac{(2\Delta t)^l}{l!} \left( -\frac{1}{2}\right)^{l-1} \Re{i^{l-1} \tr{\ad^{l-1}_{H_0}(H_1) \rho_c(\boldsymbol{\theta})Z_{\text{prod}}}}\right)\right],
\end{align}
and
\begin{align}
    &e^{\infty}_{+} - e^{\infty}_{-} = \frac{1}{\mathcal{M}_{\Delta t}} \left( b(\btheta) + 2 \Delta t \sum_{l=0}^{\infty} \frac{(2\Delta t)^l}{(l+1)!^l} \Re{(-i/2)^l\tr{\ad^{l}_{H_0}(H_1) \rho_c(\boldsymbol{\theta})Z_{\text{prod}}}}\right) =  \\ &\nonumber \hspace{3cm} \frac{1}{\mathcal{M}_{\Delta t}}\left( b(\boldsymbol{\theta}) + \Delta t \pdv{}{\theta_1}C(\boldsymbol{\theta})  \Big |_{\boldsymbol{\theta} = (1, 0)^T} \right), \hspace{1cm}
\end{align}
with the bias $b(\btheta) = \Re{\tr{\rho_c(\boldsymbol{\theta})Z_{\text{prod}}}}$ that needs to be removed. 
The variance analysis of the LCU approach shown in Table \ref{tab:samp_comp} implies a $O(\mathcal{M}_{\Delta t}^2/\epsilon^2)$ = $O \left(e^{4\abs{\Delta t}}/\epsilon^2 \right)$ asymptotic sampling complexity. For an approximation index $L$ we have $O(\mathcal{M}_L^2/\epsilon^2)$. As the de-biased estimate is multiplied with $\Delta t$, compensating for it will increase the variance by a factor $1/\Delta t^2$. As such, the optimal choice for $\Delta t$ is $O(1)$. Amplitude estimation reduces the sampling complexities by a quadratic factor. The required coherent states can be constructed efficiently with, e.g., the quantum algorithms developed in \cite{Arrazola2019, RLiu2022}.

\subsection{Quantum control gradients}\label{appx:QC_grad}
In quantum control problems, the unitary dynamics is determined by the Hamiltonian with drift operator $H_0$ and control operators $H_1,...,H_{K-1}$:
\begin{align}\label{eq:control_ham_appx}
    H_{\boldsymbol{\phi}}(t) = H_0 + \sum_{k=1}^{K-1} H_k u^{\boldsymbol{\phi}_k}_k(t).
\end{align}
The time-dependent control pulses $u^{\boldsymbol{\phi}_1}_1(t),...,u^{\boldsymbol{\phi}_{K-1}}_{K-1}$ depend on parameters $\boldsymbol{\phi}_1, ..., \boldsymbol{\phi}_{K-1}$. The solution of the control problem depends on a parameter vector $\boldsymbol{\phi} = (\boldsymbol{\phi}_1, ..., \boldsymbol{\phi}_{K-1})^{\text{T}}$. The unitary control dynamics with time-ordering operator $\mathcal{T}$ is given by:
\begin{align}\label{eq:unitary_control_appx}
    V_{\boldsymbol{\phi}}(t_0, T) = \mathcal{T} \exp{-i \int_{t=t_0}^T H_{\boldsymbol{\phi}}(\tau)d\tau}.
\end{align}
We want now to consider the case, similar to Refs.~\cite{kottmann2023evaluating, Wiersema2024herecomessun}, in which we aim to perform a variational pulse-level optimization of a quantum cost function using the unitary Eq.~\eqref{eq:unitary_control_appx}. Therefore, we need to estimate the gradient of such quantum cost functions with respect to the control dynamics.

\subsubsection{Continuous GRAPE gradient}
We study now the derivative of the cost function $C(\boldsymbol{\phi}) = \tr{V_{\boldsymbol{\phi}}(t_0,T) \rho V^{\dagger}_{\boldsymbol{\phi}}(t_0,T) \mathcal{O}}$ -- see Eq.~\eqref{eq:control_cost_function} -- with respect to the control parameters $\phi_{ij}$, where $i$ refers to the $i$th real parameter value of the vector $\boldsymbol{\phi}_j$ that corresponds to the $j$th control operator \cite{kottmann2023evaluating, Leng2024}: 
\begin{align}
    &\pdv{C(\boldsymbol{\phi})}{\phi_{ij}} = \tr{\pdv{V_{\boldsymbol{\phi}}}{\phi_{ij}} \rho V_{\boldsymbol{\phi}}^{\dagger} \mathcal{O}} + \tr{V_{\boldsymbol{\phi}} \rho \pdv{V_{\boldsymbol{\phi}}^{\dagger}}{\phi_{ij}} \mathcal{O}} = + i \int_{t_0}^T \tr{\rho(T) V_{\boldsymbol{\phi}}(t_0, \tau) H_j V_{\boldsymbol{\phi}}(t_0, \tau) \mathcal{O}} \pdv{u^{\boldsymbol{\phi}_j}_j(\tau)}{\phi_{ij}} d\tau = \\ &\hspace{1.3cm}  -i \int_{t_0}^T \nonumber \tr{\rho(T) [H_j(t_0, \tau), \mathcal{O}]} \pdv{u^{\boldsymbol{\phi}_j}_j(\tau)}{\phi_{ij}} d\tau = -i \int_{t_0}^T s_j^{\boldsymbol{\phi}}(\tau)  \pdv{u^{\boldsymbol{\phi}_j}_j(\tau)}{\phi_{ij}} d\tau,
\end{align}
where $\rho(T) = V_{\boldsymbol{\phi}}(t_0,T) \rho V_{\boldsymbol{\phi}}(T,t_0)$, $H_j(t_0,\tau) = V_{\boldsymbol{\phi}}(t_0,\tau) H_j V_{\boldsymbol{\phi}}(\tau, t_0)$ and $s_j^{\boldsymbol{\phi}}(\tau) = \tr{\rho(T) [H_j(t_0, \tau), \mathcal{O}]}$. Here we used the property of the propagator $V(t_0, T) = V(t_0, \tau)V(\tau, T)$ and the cyclic properties of the trace. We see that the gradient of the quantum cost function with respect to the pulse parameters is equal to the time integral of a different time-dependent quantum cost function multiplied with the derivative of classical time-dependent signals $u^{\boldsymbol{\phi}_j}_j(\tau)$, which can be determined using classical automatic differentiation of the input pulses coming from an amplitude waveform generator. 

\subsubsection{LCU GRAPE circuit}
Similarly as in the case of the SU($d$) gradient, the quantum control gradient consists of a (continuous) sum of estimates of quantum observables. Due to the presence of the integral, this gradient can be evaluated using either the analog or the standard LCU approach \cite{Chakraborty2024implementingany}. Another possibility is to use the GRAPE approach \cite{Khaneja2005}: the time propagation of the pulse is divided in a sequence of time samples $t_0, t_1, ..., t_{N_T-1}$ for a pulse $u(t_0), u(t_1), ..., u(t_{N_T-1})$. In this case, the gradient is given by a discrete sum of terms -- see Eq.~\eqref{eq:grape_sum} --, which can be estimated using the circuit given in Fig.~\eqref{fig:grape_circuit} (b). 
The length of the circuit grows here as $O[N_T \log(N_T)]$, which is the number of discrete time step considered to represent the continuous integral given in Eq.~\eqref{eq:grape_continuous_grad}. 

As the GRAPE circuits is particularly challenging to implement on basic quantum devices, it cannot be used to perform gate set optimization and compilation on quantum hardware. It can be implemented, however, in contexts where quantum control is used as an ansatz for variational quantum algorithms \cite{Magann2021} and on a quantum hardware platform that allows for efficient implementations of multi-qubit operations. In particular, the amplified version of this gradient allows for quadratic speed-up over naive gradient evaluation, due to the presence of a (coherent) quantum summation process instead of a classical one. 

\section{Truncation Error in the Gradient Series} \label{sec:truncation_error}

By inserting the expansion defined in Eq.~\eqref{eq:SUN_gradient_omega} into the gradient given in Eq.~\eqref{eq:potq_grad}, we obtain: 
\begin{equation}
\nabla_{\btheta}C(\btheta) = \sum_{l=0}^{\infty} T_l(\btheta),
\end{equation}
with 
\begin{equation}
T_l = \frac{(-1)}{(l+1)!} \|\btheta\|^l \frac{1}{d} \Re \Tr\bigl[(-i)^l V \ad_H^l(G) U^\dagger\bigr], \label{eq:single_order_appx}
\end{equation}
and substituting $G = \nabla_{\btheta}\bar H(\btheta)$.
For an $L$th order expansion 
\begin{equation}
\nabla_{\btheta}C_L(\btheta) = \sum_{l=0}^{L} T_l(\btheta),
\end{equation}
the remainder is then given by 
\begin{equation}
    R_L(\btheta)=\sum_{l=L+1}^\infty T_l(\btheta). \label{eq:remainder}
\end{equation}    
In the following sections we will derive bounds and estimations for the square of the truncation error $R_L$.

\section{Gradient Convergence with Random Operators} \label{sec:random_operators}
We consider the Hermitian operators $A$ and $B$, which are sampled from 2 Gaussian Unitary Ensembles (GUE) of operators, 
parametrized by the Dyson indexes $\beta_A$, $\beta_B$ and operator norm expectation values $\Lambda_A^2 = \EX\left[\Tr(A^2)\right] = \EX\left[\sum_{i=1} (\lambda_i^A)^2\right]$, 
$\Lambda_B^2 = \EX\left[\Tr(B^2)\right] = \EX\left[\sum_{i=1} \lambda_i^B\right]$, where $\lambda_i^A$, $\lambda_i^B$ are the eigenvalues of $A$, $B$ respectively. The GUE is defined by the following Gaussian measure defined on the space of $d \times d$ complex Hermitian matrices \cite{Potters_Bouchaud_2020}:
\begin{align}\label{eq:gue_def}
    \frac{1}{Z_{\text{GUE}}} e^{- \frac{d}{2} \tr{H}^2},
\end{align}
where the partition function is given by $Z_{\text{GUE}} = 2^{d/2} \left(\frac{\pi}{d}\right)^{\frac{d^2}{2}}$. In the simulations and derivations that use random Hamiltonians -- see, e.g., Fig.~\ref{fig:sun_gradient_sim} -- such Hamiltonians are sampled using Eq.~\eqref{eq:gue_def}.
We furthermore consider the unitaries to be  $U, V \sim \text{ Haar on } U(d)$, where we first treat the case of independent $U$ and $V$ to then generalize to the correlated case as required for partially optimized systems.
We are interested in deriving statistical properties of the remainder defined in Equation~\eqref{eq:remainder}. 
Using the unitary invariance of the trace, we can express the operators in the diagonal basis of $A$, so that 
\begin{equation}
    A = U \underbrace{\diag(\lambda_1, \dots , \lambda_d)}_{\tA} U^\dagger,
\end{equation}
with the diagonal operator $\tilde{A}$. In this basis $B$ is expressed as 
\begin{equation}
    \tB =  \sum_{p,q=1}^d b_{pq} \ket{p}\bra{q}.
\end{equation} 

For the first order we consider the commutator, 
\begin{equation}
    [\tA, \tB ] = \sum_{p,q=1}^d (\lambda_p - \lambda_q) b_{pq} \ket{p}\bra{q}.
\end{equation}
This generalizes to the adjoint operator recursively:
\begin{equation}
    \ad_{\tA}^l(\tB) = [\tA, \ad_{\tA}^{l-1}(\tB)] = \sum_{p,q=1}^d (\lambda_p - \lambda_q)^l b_{pq} \ket{p}\bra{q},
\end{equation}
which conveniently does not introduce additional sums.
Multiplication with the unitary operator $V$, which we also represent in the eigenbasis of $A$ as $\tV = \sum_{p,q=1}^d v_{pq} \ket{p}\bra{q}$ gives 
\begin{equation}
    \tV\ad_{\tA}^l(\tB) = \sum_{p,q,r=1}^d (\lambda_p - \lambda_q)^l  v_{rp} b_{pq} \ket{r}\bra{q}.
\end{equation}
The (real) trace of this is then equivalent to the original operators, and reduces to 
\begin{equation} \Re{\Tr\left( V \ad_{A}^l(B) \right)} = \Re{\Tr\left( \tV \ad_{\tA}^l(\tB) \right)} = \sum_{p,q=1}^d (\lambda_p - \lambda_q)^l \Re{v_{qp} b_{pq}}.
\end{equation}
Similarly for the complete remainder term, we have 
\begin{equation}
  R_L = \sum_{l=L+1}^\infty T_l = \frac{1}{d} \sum_{p,q=1}^d \Re(v_{qp} b_{pq}) \sum_{l=L+1}^\infty \frac{(-1)^l \|\btheta\|^l (\lambda_p - \lambda_q)^l }{(l+1)!}. 
\end{equation}
The expectation value of the squared remainder follows as 
\begin{equation}
  \EX[R_L^2] = \frac{1}{d^2} \sum_{p,q=1}^d \EX\left[ \Re(v_{qp} b_{qp})^2 \right] \sum_{l,m=L+1}^\infty \frac{(-1)^{l+m} \|\btheta\|^{l+m} \EX\left[(\lambda_p - \lambda_q)^{l+m}\right] }{(m+1)! (l+1)!},  \label{eq:remainder_2}
\end{equation} 
Due to circular symmetry of the Haar measure unitary, we have $\EX[v_{qp}]=\EX[v_{qp}^2] = 0$ and $\EX[\|v_{qp}\|^2] = \tfrac{1}{d}$.
For $B$ meanwhile we have $\EX[b_{qp}] = 0$ and $\EX[\|b_{qp}\|^2] = \frac{\Lambda_B^2}{d}$ (for $p \neq q$). Hence we find $\boxed{\EX\left[ \Re(v_{qp} b_{qp})^2 \right] = \frac{1}{2} \EX[\|v_{qp}\|^2] \EX[\|b_{qp}\|^2] = \frac{\Lambda_B^2}{2 d^2}}$, 
summed over $d(d-1)$ index combinations $(p,q)$ with $p \neq q$. \\

We will approximate the eigenvalue distribution without correlations, neglecting the level repulsion, using the large $d$ limiting case known as the Wigner semi-circle law \cite{Wigner1958semicircle},
so that the eigenvalue distribution is approximated as
\begin{equation}
  \rho(\lambda) = \frac{2}{\pi R^2} \sqrt{R^2 - \lambda^2} \quad \text{for } |\lambda| \leq R,  
\end{equation}
with $R = 2 \Lambda_A$. 
Due to the symmetry of the Wigner semi-circle law, the odd moments vanish $\EX[\lambda^{2m+1}] = 0$.
for the even moments we have
\begin{equation}
  \EX[\lambda^{2m}] = \frac{1}{m+1}\left(\frac{R_H}{2}\right)^{2m} \binom{2m}{m} = C_n \Lambda_A^{2m},
\end{equation}
with the Catalan numbers $C_n = \frac{1}{n+1}\binom{2n}{n}$. Hence we find using binomial expansion 
\begin{equation}
  \EX[(\lambda_p-\lambda_q)^{2m}] = \sum_{l=0}^{2m} \binom{2m}{l} (-1)^l \EX[\lambda_p^{l}] \EX[\lambda_q^{2m-l}] = \Lambda_A^{2m} \sum_{l=0}^m \binom{2m}{2l} C_l C_{m-l},
\end{equation}
with all odd terms vanishing. \\

We hence find for the remainder in Equation~\eqref{eq:remainder_2}, using the substitution $2k = l+m$
\begin{equation}
  \EX[R_L^2] = \frac{\Lambda_B^2 d(d-1)}{2 d^4} \sum_{k=L+1}^\infty \left( \|\btheta\|\Lambda_A \right)^{2k} \sum_{l=0}^k \binom{2k}{2l} C_l C_{k-l} \sum_{m=L+1}^{2k-L-1} \frac{1}{(m+1)! (2k-m+1)!}. \label{eq:exact_RL2}
\end{equation}
To leading order (i.e., setting $k= L+1$) the final sum contributes only a single term $1/(L+2)!^2$, so that 
\begin{equation}\label{eq:approx_RL2}
  \EX[R_L^2] \approx \frac{\Lambda_B^2 (d-1)}{2 d^3} \left(\|\btheta\|\Lambda_A \right)^{2(L+1)} \sum_{l=0}^{L+1} \frac{\binom{2(L+1)}{2l} C_l C_{L+1-l} }{(L+2)!^2}.
\end{equation}

\subsection{Generalizing the result for the infidelity}
The gradient of the infidelity, which is computed with the Hilbert-Schmidt test, is given by:
\begin{align}\label{eq:grad_fidelity_appendix}
    \nabla_{\btheta} C(\btheta) = - \frac{2}{d^2} \Re{\tr{\left( V(\btheta)U^{\dagger} \right) \otimes \left( \nabla_{\btheta}V(\btheta)U^{\dagger} \right) }}.
\end{align}
By using the SU($d$) derivative:
\begin{align}\label{eq:omega_fidelity_appendix}
    \nabla_{\btheta}V(\btheta) = \sum_{l=0}^{\infty} \frac{(-i)^l}{(l+1)!} \norm{\btheta}_1^l \ad^{l}_{\bar{H}(\btheta)} \left(\bar{\nabla}_{\btheta} H(\btheta) \right) V(\btheta) = \sum_{l=0}^\infty \mathcal{W}_l,
\end{align}
so that 
\begin{equation}
    \nabla_{\btheta} C(\btheta) = \frac{2}{d} \sum_{l=0}^\infty \Re{\tr{\left( V(\btheta)U^{\dagger} \right) \otimes \left( \mathcal{W}_l U^{\dagger} \right) }} = \sum_{l=0}^\infty \mathcal{T}_l.
 \end{equation}

We now turn to the squared truncation error of the gradient of the (in)fidelity, which is given in Eq.~\eqref{eq:fidelity_grad}.  We denote this truncation error with $\mathcal{R}_L = \sum_{l=L+1}^\infty \mathcal{T}_l$. The procedure is the same as for the gradient of the POTQ circuit. In this case, Eq.~\eqref{eq:remainder_2} becomes:
\begin{equation}
  \EX[\mathcal{R}_L^2] = \frac{1}{d^4} \sum_{p,q,s =1}^d \EX\left[ \Re(v_{ss} \, v_{qp} b_{qp} )^2 \right] \sum_{l,m=L+1}^\infty \frac{(-1)^{l+m} \|\btheta\|^{l+m} \EX\left[(\lambda_p - \lambda_q)^{l+m}\right] }{(m+1)! (l+1)!},  \label{eq:remainder_2_I},
\end{equation}
where we reuse $[V(\btheta) U^\dagger]_{p,q} = v_{pq}$.
Using column phase invariance, the expectation values expand as 
\begin{align}
\EX \left[ \Re(v_{ss} \, v_{qp} b_{qp} )^2 \right] = \frac{1}{2} \EX[ \|v_{ss}\|^2 \, \|v_{qp}\|^2 \|b_{qp}\|^2 ] = \frac{1}{2} \EX[ \|v_{ss}\|^2 \, \|v_{qp}\|^2] \, \EX[ \|b_{qp}\|^2 ].
\end{align} 
The $v$-dependent term contains correlations between the components, due to the column normalization, which can easily be derived via the fourth order moments for Haar unitaries, so that 
\begin{align}
    \EX[ \|v_{ss}\|^2 \, \|v_{qp}\|^2] = \begin{cases}
        \frac{1}{d(d+1)} & \text{for} \quad s\in {p,q} \; q=s \\
        \frac{1}{d^2 - 1} & \text{for} \quad s\notin {p,q} 
    \end{cases}.
\end{align}
Hence, with $\EX[ \|b_{qp}\|^2 ]=\frac{\Lambda_B^2}{d}$, we find for the sum
\begin{equation}
    \sum_{p,q,s =1}^d \EX\left[ \Re(v_{ss} \, v_{qp} b_{qp} )^2 \right]  = \frac{1}{2} \sum_{p,q =1}^d\underbrace{\left(\frac{2}{d(d+1)}+ \frac{d-2}{d^2-1}\right)}_{=\frac{d^2-2}{d^2(d-1)}} \EX[ \|b_{qp}\|^2 ] = \frac{d^2-2}{2 d^3} \Lambda_B^2.
\end{equation}

The squared truncation error of the fidelity gradient can therefore be approximated asymptotically as:
\begin{align}
    \EX[\mathcal{R}_L^2] = \frac{4(d^2-2)}{d^3(d-1)} \EX[R_L^2] \approx \frac{4}{d^2} \EX[R_L^2].   
\end{align} 

\subsection{Partially optimized $U$}
When using the gradients to optimize $V$, we can no longer assume random $V$. In the following section we discuss the changes to our estimation theory introduced by these correlations for the most pertinent case, the minimization of the infidelity between $U$ and $V(\btheta)$. With a fidelity $F= \tfrac{1}{d^2} \left|\tr\left( V(\btheta) U^\dagger \right)\right|^2= \tfrac{1}{d^2} \left|\tr\left( W \right)\right|^2$ we can decompose $W= \sum_{i=0}^{d^2-1} w_i P_i$, where $P_i$ are the $n$ qubit normalized Pauli matrices, with $P_0 = I$ and the $w_i$ are the Pauli expansion coefficients, so that 
$|w_0|^2=F$ and $\sum_{i=1}^{d^2-1} |w_i|^2= 1-F$. We assume that the remaining $d^2-1$ components, that are not constrained by the fidelity, can still be considered as randomly sampled according to the Haar distribution, and only constrained to the reduced normalization $1-F$. In order to compute the resulting expectation value, let us rewrite the gradient of the infidelity from Eq.~\eqref{eq:grad_fidelity_appendix} as 
\begin{align}\label{eq:fid_grad_no_tensor}
    \grad_{\btheta} I(\btheta) = -\frac{2}{d} \Re\Big[ \underbrace{\frac{1}{d}\tr(V(\btheta) U^\dagger)^*}_{=w_0^*} \tr(\grad_{\btheta} V(\btheta) U^\dagger) \Big],
\end{align}
where we used the property of the trace $\tr{A \otimes B} = \tr{A} \tr{B}$ for complex matrices $A, B$.
We use again the expansion:
\begin{equation}
    \nabla_{\btheta}V
=\sum_{l=0}^\infty \frac{(-i)^l}{(l+1)!} \|\btheta\|^l\,
\ad^l_{\bar H(\btheta)}\!\big(\nabla_{\btheta}\bar H(\btheta)\big)\,V,
\end{equation}
so that
\begin{align}\label{eq:fid_grad_no_tensor_expanded}
    \grad_{\btheta} I(\btheta) = -\frac{2}{d} \sum_{l=0}^\infty \frac{1}{(l+1)!} \Re\Big[ w_0^* \; (-i)^l \|\btheta\|^l\,
\tr\left(\ad^l_{\bar H(\btheta)}\!\big(\nabla_{\btheta}\bar H(\btheta)\big)\,W \right) \Big] = -\frac{2}{d} \sum_{l=0}^\infty \mathcal{I}_l.
\end{align}
The trace of a product of matrices represented in the (normalized) Pauli expansion picture, i.e., $A = \sum_{i=1}^{d^2 - 1} a_i P_i$ and $B = \sum_{j=1}^{d^2 -1} b_j P_j$, where $P_i, P_j$ are normalized $n$-qubit Pauli strings, is given by:
\begin{equation}
    \tr(A B) = \sum_{i=0}^{d^2-1} a_i b_i.
\end{equation}
The adjoint operators are composed of commutators, which are traceless, and hence the first Pauli component of the adjoint vanishes.  
By separating $W= W_0+W_r$, so that $W_0= w_0 P_0$ and $W_r = \sum_{i=1}^{d^2 -1} w_i P_i$, we can transform 
\begin{equation}
    \tr\left(\ad^l_{\bar H(\btheta)}\!\big(\nabla_{\btheta}\bar H(\btheta)\big)\,W \right) = \tr\left(\ad^l_{\bar H(\btheta)}\!\big(\nabla_{\btheta}\bar H(\btheta)\big)\,W_r \right) = \sum_{p,q=1}^d (\lambda_p - \lambda_q)^l v_{qp} b_{pq},
\end{equation}
where we reuse the previous notation, where $B= \grad_{\btheta} \bar{H}(\btheta)$, but with the reduced amplitude $[W_r]_{qp} = v_{qp}$. We know, that $W=w_0 I + i \sqrt{\epsilon} A$, with $|w_0|^2 = F$, $\epsilon= 1- F$ and identifying $W_0 = w_0 I$, $W_r = i\sqrt{\epsilon} A$, where $A$ is unitary. We are then left with $\EX[|v_{pq}|^2]= \frac{\epsilon}{d} = \frac{1-F}{d}$. Hence, for the correlated remainder $\mathcal{R}_l = \sum_{l=L+1}^\infty \mathcal{I}_l$, the expected value over random Hamiltonians $H_0$ and $H_1$ reads:
\begin{align}
  \EX[R_L^2] &= \frac{4 F(\btheta)}{d^2} \sum_{p,q=1}^d \EX[ |v_{qp}|^2] \EX[|b_{qp}|^2 ] \sum_{l,m=L+1}^\infty \frac{(-1)^{l+m} \|\btheta\|^{l+m} \EX\left[(\lambda_p - \lambda_q)^{l+m}\right] }{(m+1)! (l+1)!}\nonumber \\
  &= \frac{4 F(\btheta)}{d^2} d(d-1) \frac{\left[ 1-F(\btheta) \right]\Lambda_B^2}{2 d^2} \sum_{l,m=L+1}^\infty \frac{(-1)^{l+m} \|\btheta\|^{l+m} \EX\left[(\lambda_p - \lambda_q)^{l+m}\right] }{(m+1)! (l+1)!} \nonumber \\
&=  \frac{2 (d-1) [F(\btheta)-F^2(\btheta)]\Lambda_B^2 }{d^3}  \sum_{l,m=L+1}^\infty \frac{(-1)^{l+m} \|\btheta\|^{l+m} \EX\left[(\lambda_p - \lambda_q)^{l+m}\right] }{(m+1)! (l+1)!}.\label{eq:expected_remainder_correlated_I}
\end{align} %

\end{document}